\documentclass[11pt,a4paper]{article}
\usepackage{amsmath,amsthm,amsfonts,amssymb,mathrsfs}
\usepackage{graphicx}
\usepackage{subfigure}
\usepackage[latin1]{inputenc}
\usepackage[swedish,english]{babel}
\newcommand{\myvec}[1]{{\boldsymbol#1}}
\newcommand{\Arg}{{\rm{Arg}}}
\newtheorem{remark}{Remark}[section]
\newtheorem{theorem}{Theorem}[section]
\newtheorem{proposition}{Proposition}[section]
\newtheorem{corollary}{Corollary}[section]

\begin{document}
\title{An extended charge-current formulation of the electromagnetic
  transmission problem} 
\author{Johan Helsing\thanks{Centre for Mathematical Sciences, Lund
    University, Sweden}~~and Anders Karlsson\thanks{Electrical and
    Information Technology, Lund University, Sweden}} 

\date{February 15, 2020} 
\maketitle

\begin{abstract}
  A boundary integral equation formulation is presented for the
  electromagnetic transmission problem where an incident
  electromagnetic wave is scattered from a bounded dielectric object.
  The formulation provides unique solutions for all combinations of
  wavenumbers in the closed upper half-plane for which Maxwell's
  equations have a unique solution. This includes the challenging
  combination of a real positive wavenumber in the outer region and an
  imaginary wavenumber inside the object. The formulation, or variants
  thereof, is particularly suitable for numerical field evaluations as
  confirmed by examples involving both smooth and non-smooth objects.
\end{abstract}

\section{Introduction}
\label{sec:introduction}

This work is about transmission problems. A simply connected
homogeneous isotropic object is located in a homogeneous isotropic
exterior region. A time harmonic incident wave, generated in the
exterior region, is scattered from the object. The aim is to evaluate
the fields in the interior and exterior regions.

We present boundary integral equation (BIE) formulations for the
solution of the scalar Helmholtz and the electromagnetic Maxwell
transmission problems. We show that our integral equations have unique
solutions for all wavenumbers $k_1$ of the exterior domain and $k_2$
of the object with $0\leq\Arg(k_1),\Arg(k_2)<\pi$, and for which the
partial differential equation (PDE) formulations of the two problems
have unique solutions. As we understand it, there is no other BIE
formulation of the electromagnetic problem known to the computational
electromagnetics community that can guarantee unique solutions for the
wavenumber combination
\begin{equation}
\Arg(k_1)=0\,, \quad \Arg(k_2)=\pi/2\,, \quad \mbox{and} \quad
k_2^2/k_1^2\ne -1\,. 
\label{eq:plasmonic}
\end{equation}

We refer to the combination \eqref{eq:plasmonic} as the {\it plasmonic
  condition} since it enables discrete quasi-electrostatic surface
plasmons in smooth, infinitesimally small, objects \cite{TzarSihv18},
continuous spectra of quasi-electrostatic surface plasmons in
non-smooth objects \cite{HelsPerf18}, and undamped surface plasmon
waves along planar surfaces \cite[Appendix I]{Raeth88}. Wavenumbers
with $\Arg(k_1)=0$ and $\pi/4<\Arg(k_2)\leq\pi/2$ are of special
interest in the areas of nano-optics and metamaterials because in this
range weakly damped surface plasmons in subwavelength objects and
weakly damped dynamic surface plasmon waves in objects of any size can
occur. These phenomena become increasingly pronounced, and useful in
applications, as $\Arg(k_2)$ approaches
$\pi/2$~\cite{Homola08,LuTsGuHo19}. It is important to have uniqueness
under the plasmonic condition, despite that there are no known
materials that satisfy this condition exactly, since non-uniqueness
implies spurious resonances that deteriorate the accuracy of the
numerical solution also for $\Arg(k_1)=0$, $\pi/4<\Arg(k_2)<\pi/2$.

It is relatively easy to find a BIE formulation of the scalar
transmission problem since one has access to the fundamental solution
to the scalar Helmholtz equation. It remains to make sure that the
boundary conditions are satisfied and that the solution is unique. To
find a BIE formulation of the electromagnetic transmission problem,
based on the same fundamental solution, is harder. Apart from
satisfying the boundary conditions and uniqueness one also has to make
sure that the solution satisfies Maxwell's equations. Otherwise the
two problems are very similar.

Our BIE formulation of the scalar problem is a modification of the
formulation in \cite[Section 4.2]{KleiMart88}. While our formulation
guarantees unique solutions under the plasmonic condition, provided
that the object surface is smooth, the formulation in \cite[Section
4.2]{KleiMart88} does not.

Our BIE formulation of the electromagnetic problem is a further
development of the classic formulation by Müller, \cite[Section
23]{Muller69}. In \cite{MautHarr77} it is shown that the Müller
formulation has unique solutions for $0\leq
\Arg(k_1),\Arg(k_2)<\pi/2$, but as shown in \cite{HelsKarl19}, it may
have spurious resonances under the plasmonic condition. The Müller
formulation has four unknown scalar surface densities, related to the
equivalent electric and magnetic surface current densities, and that
leads to dense-mesh/low-frequency breakdown in field evaluations.
Despite these shortcomings, the Müller formulation has been frequently
used. Its advantages are emphasized in a recent paper
\cite{LaiOneil19} on scattering from axisymmetric objects where
accurate solutions are obtained away from the low-frequency limit.

One way to overcome low-frequency breakdown in the Müller formulation
is to increase the number of unknown densities from four to six by
adding the equivalent electric and magnetic surface charge densities
\cite{HelsKarl17,TaskOija06,VicGreFer18}. The charge densities can be
introduced in two ways, leading to two types of formulations. The
first type is decoupled charge-current formulations, where the charge
densities are introduced after the BIE has been solved. The other type
is coupled charge-current formulations, where the charge densities are
present from the start. Unfortunately, both types of formulations can
give rise to new complications such as spurious resonances and
near-resonances. Several formulations in the literature ignore these
complications, but in \cite{VicGreFer18} a stable formulation is
presented. In line with all other formulations in literature,
uniqueness in \cite{VicGreFer18} is not guaranteed under the plasmonic
condition.

The main result of the present work is our extended charge-current BIE
formulation of the electromagnetic transmission problem where two
additional surface densities, related to electric and magnetic volume
charge densities, are introduced. The formulation is given by the
representation~(\ref{eq:repEH}) and the system~(\ref{eq:EHsys}) below.
The formulation is free from low-frequency breakdown and it provides
unique solutions also under the plasmonic condition. Just like the
Müller- and charge-current formulations it is a direct formulation,
meaning that the surface densities are related to boundary limits of
fields, or derivatives of fields. This is in contrast to indirect
formulations \cite{EpsGreNei13,EpsGreNei19,KresRoac78,VicGreFer18},
where the surface densities lack immediate physical interpretation.
Albeit somewhat more numerically expensive than competing
formulations, our new formulation enables high achievable accuracy and
since it is more robust this should outweigh the disadvantage of
having eight densities. From a broader perspective one can say that
our paper, and many other papers
\cite{HelsKarl17,LaiOneil19,MautHarr77,Muller69,TaskOija06,VicGreFer18},
use integral representations of the electric and magnetic fields for
modeling. It is also possible to start with representations of scalar
and vector potentials and antipotentials
\cite{EpsGreNei13,EpsGreNei19,LiFuShank18}.

The paper is organized as follows: Section \ref{sec:pre} introduces
notation and definitions common to the scalar and the electromagnetic
problems. The scalar problem and two closely related homogeneous
problems, to be used in a uniqueness proof, are defined in
Section~\ref{sec:scalarprob}. Scalar integral representations
containing two surface densities are introduced in
Section~\ref{sec:intrepA}. Section~\ref{sec:inteqA} proposes a system
of BIEs for these densities. This system contains two free parameters
and, as seen in Section~\ref{sec:uniqueA}, unique solutions are
guaranteed by giving them proper values. Section~\ref{sec:eval}
concerns the evaluation of near fields. The procedure for finding BIEs
for the scalar problem is then adapted to the electromagnetic problem,
defined along with two auxiliary homogeneous problems in
Section~\ref{sec:EMprob}. Integral representations of electric and
magnetic fields in terms of eight scalar surface densities are given
in Section~\ref{sec:intrepC} and a corresponding system of BIEs is
proposed in Section~\ref{sec:inteqC}. This BIE system contains four
free parameters and again, as shown in Section~\ref{sec:uniqueC},
unique solutions are guaranteed by choosing them properly.
Section~\ref{sec:twoD} presents reduced two-dimensional (2D) versions
of the electromagnetic BIE system whose purpose is to facilitate
initial tests and comparisons. Section~\ref{sec:geomdisc} reviews test
domains and discretization techniques and Section~\ref{sec:numex}
presents numerical examples, including what we believe is the first
high-order accurate computation of a surface plasmon wave on a
non-smooth three-dimensional (3D) object.

Appendix A presents boundary values of integral representations.
Appendix B and C derive conditions for our representations of the
electric and magnetic fields to satisfy Maxwell's equations. In
Appendix D a set of points $(\Arg(k_1),\Arg(k_2))$ is identified for
which the electromagnetic problem has at most one solution.

\begin{figure}[t]
\centering
\includegraphics[height=50mm]{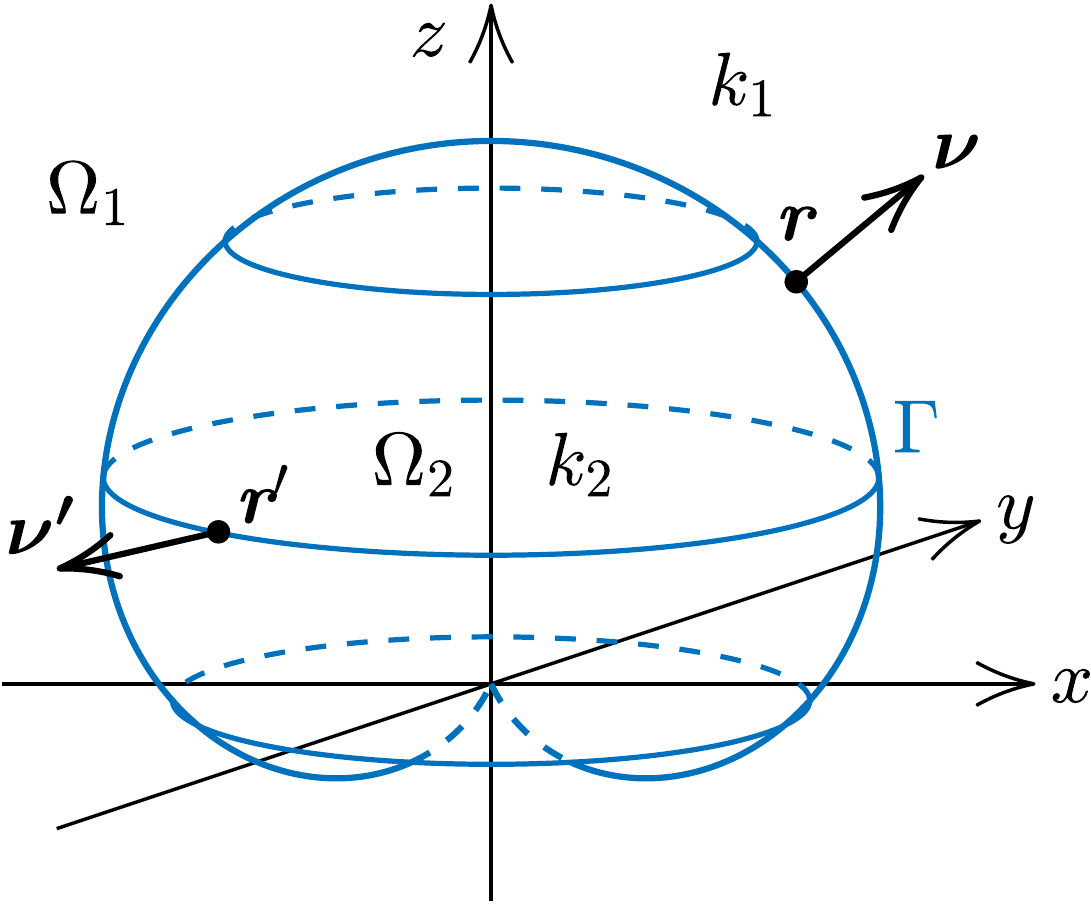}
\caption{\sf Geometry in $\mathbb{R}^3$. Inside $\Gamma$
  the volume is $\Omega_2$ and the wavenumber $k_2$. Outside $\Gamma$
  the volume is $\Omega_1$ and the wavenumber $k_1$. The outward unit
  normal is $\myvec \nu$ at $\myvec r$ and $\myvec \nu'$ at $\myvec
  r'$.}
\label{fig:amoeba0}
\end{figure}

\section{Notation}
\label{sec:pre}

Let $\Omega_2$ be a bounded volume in $\mathbb{R}^3$ with a smooth
closed surface $\Gamma$ and simply connected unbounded exterior
$\Omega_1$. The outward unit normal at position $\myvec r$ on $\Gamma$
is $\myvec\nu$. We consider time-harmonic fields with time dependence
$e^{-{\rm i}t}$, where the angular frequency is scaled to one. The
relation between time-dependent fields $F(\myvec r,t)$ and complex
fields $F(\myvec r)$ is
\begin{equation}
F(\myvec r,t)=\Re{\rm e}\left\{F(\myvec r)e^{-{\rm i}t}\right\}. 
\label{eq:timedep}
\end{equation}
The volumes $\Omega_1$ and $\Omega_2$ are homogeneous with wavenumbers
$k_1$ and $k_2$. See Figure \ref{fig:amoeba0}, which depicts a
non-smooth $\Gamma$ that is used later in numerical examples. An
incident field is generated by a source somewhere in $\Omega_1$.

\subsection{Layer potentials and boundary integral operators}
\label{sec:laypot}

The fundamental solution to the scalar Helmholtz equation is taken to
be
\begin{equation}
\Phi_k(\myvec r,\myvec r')=
\frac{e^{{\rm i}k\lvert\myvec r-\myvec r'\rvert}}
{4\pi\lvert\myvec r-\myvec r'\rvert}\,,\quad
\myvec r,\myvec r'\in\mathbb{R}^3\,.
\label{eq:fund}
\end{equation} 
Two scalar layer potentials are defined in terms of a general surface
density $\sigma$ as
\begin{equation}
\begin{split}
S_k\sigma(\myvec r)&=2\int_{\Gamma}\Phi_{k}(\myvec r,\myvec r')
\sigma(\myvec r')\,{\rm d}\Gamma'\,,
\quad\myvec r\in\Omega_1\cup\Omega_2\,,\\
K_k\sigma(\myvec r)&=2\int_{\Gamma}(\partial_{\nu'}\Phi_{k})
(\myvec r,\myvec r')\sigma(\myvec r')\,{\rm d}\Gamma'\,,
\quad\myvec r\in\Omega_1\cup\Omega_2\,,
\end{split}
\label{eq:STk}
\end{equation}
where ${\rm d}\Gamma$ is an element of surface area,
$\partial_{\nu'}=\myvec\nu'\cdot\nabla'$, and
$\myvec\nu'=\myvec\nu(\myvec r')$. We use~(\ref{eq:STk}) also for
$\myvec r\in\Gamma$, in which case $S_k$ and $K_k$ are viewed as
boundary integral operators. Further, we need the operators $K_k^{\rm
  A}$ and $T_k$, defined by
\begin{equation}
\begin{split}
K^{\rm A}_k\sigma(\myvec r)&=2\int_{\Gamma}
(\partial_\nu\Phi_{k})(\myvec r,\myvec r')\sigma(\myvec r')
  \,{\rm d}\Gamma'\,,
\quad\myvec r\in\Gamma\,,\\
T_k\sigma(\myvec r)&=2\int_{\Gamma}
(\partial_\nu\partial_{\nu'}\Phi_{k})
(\myvec r,\myvec r')\sigma(\myvec r')\,{\rm d}\Gamma'\,,
\quad\myvec r\in\Gamma,
\end{split}
\end{equation}
and where $T_k\sigma$ is to be understood in the Hadamard finite-part
sense. We also need the vector-valued layer potentials
\begin{equation}
\begin{split}
{\cal S}_k\myvec \sigma(\myvec r)&=2\int_\Gamma
\Phi_{k}(\myvec r,\myvec r')\myvec \sigma(\myvec r')\,{\rm d}\Gamma'\,,
\quad\myvec r\in\Omega_1\cup\Omega_2\,,\\
\boldsymbol{\cal N}_k\sigma(\myvec r)&=2\int_\Gamma
\nabla\Phi_{k}(\myvec r,\myvec r')\sigma(\myvec r')\,{\rm d}\Gamma'\,,
\quad\myvec r\in\Omega_1\cup\Omega_2\,,\\
{\cal R}_k\myvec \sigma(\myvec r)&=2\int_{\Gamma}\nabla\Phi_{k}
(\myvec r,\myvec r')\times\myvec \sigma(\myvec r')\,{\rm d}\Gamma'\,,
\quad\myvec r\in\Omega_1\cup\Omega_2\,,
\end{split}
\label{eq:cSNR}
\end{equation}
with corresponding operators ${\cal S}_k$, $\boldsymbol{\cal N}_k$,
and ${\cal R}_k$ for $\myvec r\in\Gamma$. The notation
\begin{equation}
\tilde S_{k}={\rm i}k_1 S_k\,,\quad
\tilde{\cal S}_{k}={\rm i}k_1 {\cal S}_k\,,
\label{eq:barS}
\end{equation}
will be used for brevity.

The fundamental solution $\Phi_k$ and the operators $S_k$, $K_k$,
$K^{\rm A}_k$, and $T_k$ are identical to the corresponding constructs
in \cite[Eqs.~(2.1) and (3.8)--(3.11)]{ColtKres98}. The potentials of
\eqref{eq:cSNR} correspond to the potentials in \cite[Eqs.~(3)
and~(9)]{TaskOija06}, scaled with a factor of two.

\subsection{Limits of layer potentials}
\label{sec:limits}

It is convenient to introduce the notation
\begin{equation}
\begin{split}
A^+(\myvec r^\circ)&=\lim_{\Omega_1\ni\myvec r\to \myvec r^\circ}A(\myvec r)\,,
\quad\myvec r^\circ\in\Gamma\,,\\
A^-(\myvec r^\circ)&=\lim_{\Omega_2\ni\myvec r\to \myvec r^\circ}A(\myvec r)\,,
\quad\myvec r^\circ\in\Gamma\,,
\end{split}
\end{equation}
for limits of a function $A(\myvec r)$ as
$\Omega_1\cup\Omega_2\ni\myvec r\to\myvec r^\circ\in\Gamma$. For
compositions of operators and functions, square brackets $[\cdot]$
indicate parts where limits are taken. In this notation, results from
classical potential theory on limits of layer potentials
include~\cite[Theorem~3.1]{ColtKres98}
and~\cite[Theorem~2.23]{ColtKres83}
\begin{equation}
\begin{split}
[S_k\sigma]^\pm(\myvec r)&=S_k\sigma(\myvec r)\,,\quad \myvec r\in\Gamma\,,\\
[K_k\sigma]^\pm(\myvec r)&=\pm\sigma(\myvec r)+K_k\sigma(\myvec r)\,,
\quad \myvec r\in\Gamma\,,\\
\myvec\nu\cdot[\nabla S_k\sigma]^\pm(\myvec r)&=
\mp\sigma(\myvec r)+K_k^{\rm A}\sigma(\myvec r)\,,
\quad \myvec r\in\Gamma\,,\\
\myvec\nu\cdot[\nabla K_k\sigma]^\pm(\myvec r)&=
T_k\sigma(\myvec r)\,,\quad\myvec r\in\Gamma\,.
\end{split}
\label{eq:SKATlim}
\end{equation}
See also~\cite[Theorem 5.46]{KirsHett15} for statements on the second
and fourth limit of~(\ref{eq:SKATlim}) in a more modern function-space
setting.

The layer potentials of~(\ref{eq:cSNR}) have limits
\begin{equation}
\begin{split}
[{\cal S}_k\myvec\sigma]^\pm(\myvec r)&=
{\cal S}_k\myvec\sigma(\myvec r)\,,\quad \myvec r\in\Gamma\,,\\
\myvec\nu\cdot[\boldsymbol{\cal N}_k\sigma]^\pm(\myvec r)&=
\mp\sigma(\myvec r)+\myvec\nu\cdot\boldsymbol{\cal N}_k\sigma(\myvec r)\,,
\quad \myvec r\in\Gamma\,,\\
\myvec\nu\times[{\cal R}_k\myvec\sigma]^\pm(\myvec r)&=
\pm\myvec\sigma(\myvec r)+\myvec\nu\times{\cal R}_k\myvec\sigma(\myvec r)\,,
\quad \myvec r\in\Gamma\,.
\end{split}
\end{equation}

\section{Scalar transmission problems}
\label{sec:scalarprob}

We present three scalar transmission problems called problem {\sf A},
problem ${\sf A}_0$, and problem ${\sf B}_0$. Problem {\sf A} is the
problem of main interest. Problem ${\sf A}_0$ and ${\sf B}_0$ are
needed in proofs.

\subsection{Problem {\sf A} and ${\sf A}_0$}
\label{sec:A0}

The transmission problem {\sf A} reads: Given an incident field
$U^{\rm in}$, generated in $\Omega_1$, find the total field $U(\myvec
r)$, $\myvec r\in\Omega_1\cup\Omega_2$, which, for a complex jump
parameter $\kappa$ and for wavenumbers $k_1$ and $k_2$ such that
\begin{equation}
0\leq\Arg(k_1), \Arg(k_2)<\pi\,,
\label{eq:k1k2}
\end{equation}
solves
\begin{equation}
\begin{split}
\Delta U(\myvec r)+k_1^2U(\myvec r)&=0\,,\quad\myvec r\in\Omega_1\,,\\
\Delta U(\myvec r)+k_2^2U(\myvec r)&=0\,,\quad\myvec r\in\Omega_2\,,
\end{split}
\label{eq:helmA12}
\end{equation}
except possibly at an isolated point in $\Omega_1$ where the source of
$U^{\rm in}$ is located, subject to the boundary conditions
\begin{align}
U^+(\myvec r)&=U^-(\myvec r)\,,\quad\myvec r\in\Gamma\,,
\label{eq:BCA1}\\
\kappa\myvec\nu \cdot[\nabla U]^+(\myvec r)&=
\myvec\nu\cdot[\nabla U]^-(\myvec r)\,,\quad \myvec r\in\Gamma\,,
\label{eq:BCA2}\\
\left(\partial_{\hat{\myvec r}}-{\rm i}k_1\right)
U^{\rm sc}(\myvec r)&=o\left(\lvert\myvec r\rvert^{-1}\right),
\quad\lvert\myvec r\rvert\rightarrow\infty\,.
\label{eq:BCA3}
\end{align}
Here $\hat{\myvec r}=\myvec r/|\myvec r|$, the scattered field $U^{\rm
  sc}$ is source free in $\Omega_1$ and given by
\begin{equation}
U(\myvec r)=U^{\rm in}(\myvec r)+U^{\rm sc}(\myvec r)\,,
\quad\myvec r\in\Omega_1\,,
\end{equation}
and the incident field satisfies
\begin{equation}
\Delta U^{\rm in}(\myvec r)+k_1^2U^{\rm in}(\myvec r)=0\,, 
\quad\myvec r\in\mathbb{R}^3\,,
\end{equation}
except at the possible isolated source point in $\Omega_1$.

The homogeneous version of problem {\sf A}, that is problem {\sf A}
with $U^{\rm in}$=0, is referred to as problem ${\sf A}_0$.

\subsection{Problem ${\sf B}_0$}\label{sec:B0}

The transmission problem ${\sf B}_0$ reads: Find $W(\myvec r)$,
$\myvec r\in\Omega_1\cup\Omega_2$, which, for a complex jump parameter
$\alpha$ and for wavenumbers $k_1$ and $k_2$ such that~(\ref{eq:k1k2})
holds, solves
\begin{equation}
\begin{split}
\Delta W(\myvec r)+k_2^2W(\myvec r)&=0\,,\quad\myvec r\in\Omega_1\,,\\
\Delta W(\myvec r)+k_1^2W(\myvec r)&=0\,,\quad\myvec r\in\Omega_2\,,
\end{split}
\label{eq:helmB12}
\end{equation}
subject to the boundary conditions
\begin{align}
W^+(\myvec r)&=W^-(\myvec r)\,,
\quad \myvec r\in\Gamma\,,
\label{eq:BCB1}\\
\alpha\myvec\nu\cdot[\nabla W]^+(\myvec r)&=
\myvec\nu\cdot[\nabla W]^-(\myvec r)\,,
\quad \myvec r\in\Gamma\,,
\label{eq:BCB2}\\
\left(\partial_{\hat{\myvec r}}-{\rm i}k_2\right)
W(\myvec r)&=o\left(\lvert\myvec r\rvert^{-1}\right),
\quad\lvert\myvec r\rvert\rightarrow\infty\,.
\label{eq:BCB3}
\end{align}

\subsection{Uniqueness and existence}
\label{sec:uniqueexist}

We now review uniqueness theorems by Kress and Roach~\cite{KresRoac78}
and Kleinman and Martin~\cite{KleiMart88} for solutions to problem
{\sf A}, along with corollaries for problem ${\sf A}_0$ and ${\sf
  B}_0$. Propositions and corollaries apply only under conditions on
$k_1$, $k_2$, $\kappa$, and $\alpha$ that are more restrictive than
those of~(\ref{eq:k1k2}). Conjugation of complex quantities is
indicated with an overbar symbol.

\begin{proposition}
Assume that $(\ref{eq:k1k2})$ holds. Let in addition $k_1$, $k_2$,
$\kappa$, $\kappa^{-1}$ $\in \mathbb{C}\backslash 0$ be such that
\begin{equation}
\begin{split}
\Arg(k_1^2\bar k_2^2\kappa)&=
\left\{
\begin{array}{lll}
0   & \text{if} & \Re{\rm e}\{k_1\}\Re{\rm e}\{k_2\}\ge 0\,,\\
\pi & \text{if} & \Re{\rm e}\{k_1\}\Re{\rm e}\{k_2\}<0 \,,
\end{array}
\right.
\\
\Arg(k_2)&\ne 0\quad \text{if}\quad \Arg(k_1)=\pi/2\,.
\end{split}
\label{eq:KRunique}
\end{equation}
Then problem {\sf A} has at most one solution.
\label{prop:KR}
\end{proposition}
\begin{proof}
  This is~\cite[Theorem~3.1]{KresRoac78}, supplemented with a
  condition to compensate for a minor flaw in the proof. The original
  conditions in~\cite[Theorem~3.1]{KresRoac78} permit combinations of
  $k_1$, $k_2$, and $\kappa$ for which problem {\sf A} has nontrivial
  homogeneous solutions. Examples can be found with $\Arg(k_1)=\pi/2$,
  $\Arg(k_2)=0$, and $\Arg(\kappa)=\pi$, using the example for the
  sphere in~\cite[p.~1434]{KresRoac78}.
\end{proof}

\begin{proposition}
Assume that $(\ref{eq:k1k2})$ holds. Let in addition $k_1$, $k_2$,
$\kappa$, $\kappa^{-1}$ $\in \mathbb{C}\backslash 0$ be such that
\begin{equation}
0\leq\Arg(k_1\kappa)\leq \pi\quad \text{and}\quad 
0\leq\Arg(\bar k_1 k_2^2\bar\kappa)\leq\pi\,.
\label{eq:crit2KL}
\end{equation}
Then problem {\sf A} has at most one solution.
\label{prop:KM}
\end{proposition}
\begin{proof}
This is the uniqueness theorem in~\cite[p.~309]{KleiMart88}.
\end{proof}

The conditions~(\ref{eq:KRunique}) intersect with the
conditions~(\ref{eq:crit2KL}). If any of these sets of conditions
holds, then we say that \textit{the conditions of
  Proposition~\ref{prop:KR} or~\ref{prop:KM} hold}. These conditions
are sufficient for our purposes but, as pointed out in
\cite[p.~1434]{KresRoac78}, uniqueness can be established for a wider
range of conditions.

\begin{corollary}
  If the conditions of Proposition~\ref{prop:KR} or~\ref{prop:KM}
  hold, then problem ${\sf A}_0$ has only the trivial solution
  $U(\myvec r)=0$.
\label{cor:uniqueA0}
\end{corollary}
  
\begin{remark}
  In Ref.~\cite{KleiMart88}, the condition~(\ref{eq:k1k2}) is not
  directly included in the formulation of what corresponds to our
  problem {\sf A}. Instead, the condition $0\leq\Arg(k_1)<\pi$ is
  added for the problem to have at most one solution and
  $0\leq\Arg(k_2)<\pi$ is added for the existence of a unique
  solution.
\end{remark}

\begin{proposition}
Assume that
\begin{equation}
\begin{split}
\Arg(\bar k_1^2 k_2^2\alpha)&=
\left\{
\begin{array}{lll}
0   & \text{if} & \Re{\rm e}\{k_1\}\Re{\rm e}\{k_2\}\ge 0\,,\\
\pi & \text{if} & \Re{\rm e}\{k_1\}\Re{\rm e}\{k_2\}<0 \,,
\end{array}
\right.
\label{eq:KRuniqueB0}
\\
\Arg(k_1)&\ne 0\quad \text{if}\quad \Arg(k_2)=\pi/2\,,
\end{split}
\end{equation}
or
\begin{equation}
0\leq\Arg(k_2\alpha)\leq\pi\quad \text{and}\quad 
0\leq\Arg(k_1^2\bar k_2\bar\alpha)\leq\pi\,,
\label{eq:crit2KLB0}
\end{equation}
holds. Then problem ${\sf B}_0$ has only the trivial solution
$W(\myvec r)=0$.
\label{prop:uniqueB0}
\end{proposition}
\begin{proof}
  Interchange $k_1$ and $k_2$ and replace $\kappa$ by $\alpha$ in
  Proposition~\ref{prop:KR} and~\ref{prop:KM}. Then use
  Corollary~\ref{cor:uniqueA0}.
\end{proof}
If any of the sets of conditions~(\ref{eq:KRuniqueB0})
or~(\ref{eq:crit2KLB0}) holds we say that \textit{the conditions of
  Proposition~\ref{prop:uniqueB0} hold}.

\begin{figure}[t]
\centering
\includegraphics[height=50mm]{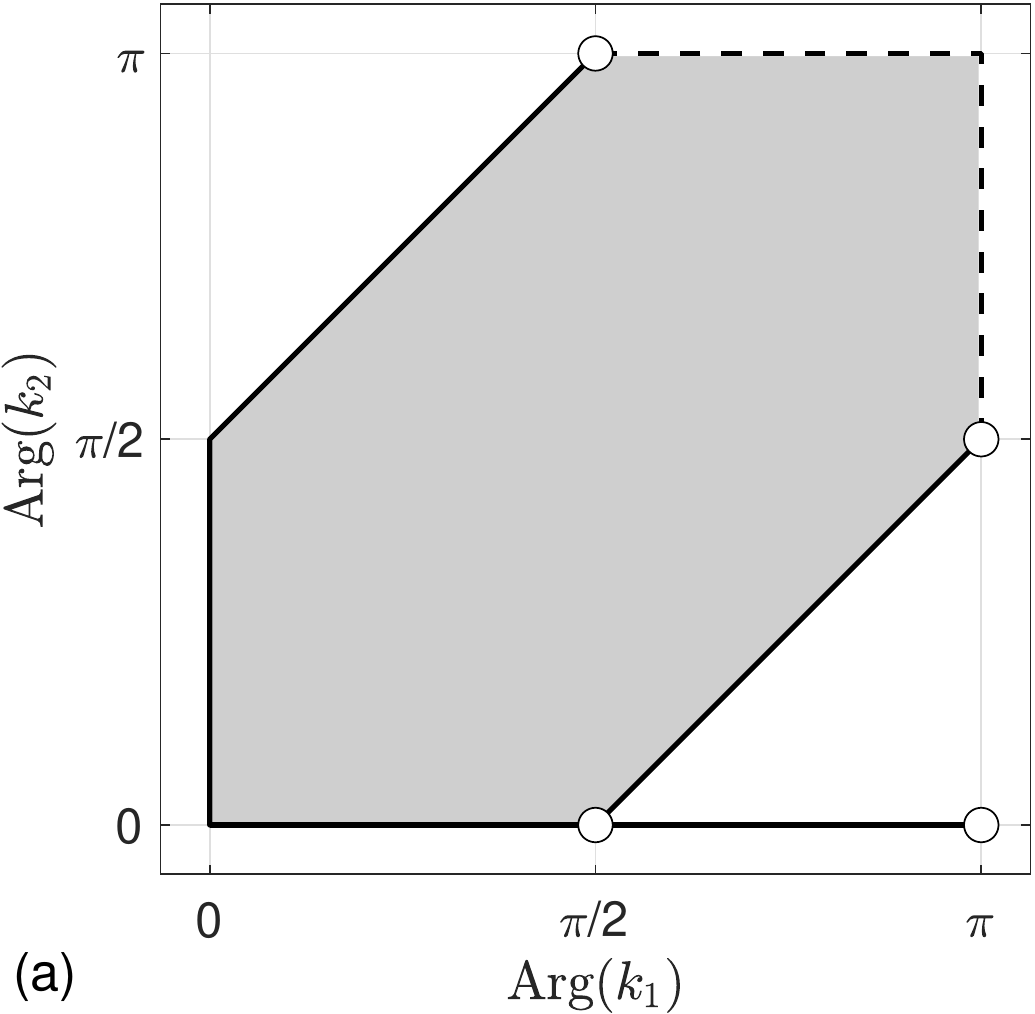}
\hspace*{5mm}
\includegraphics[height=50mm]{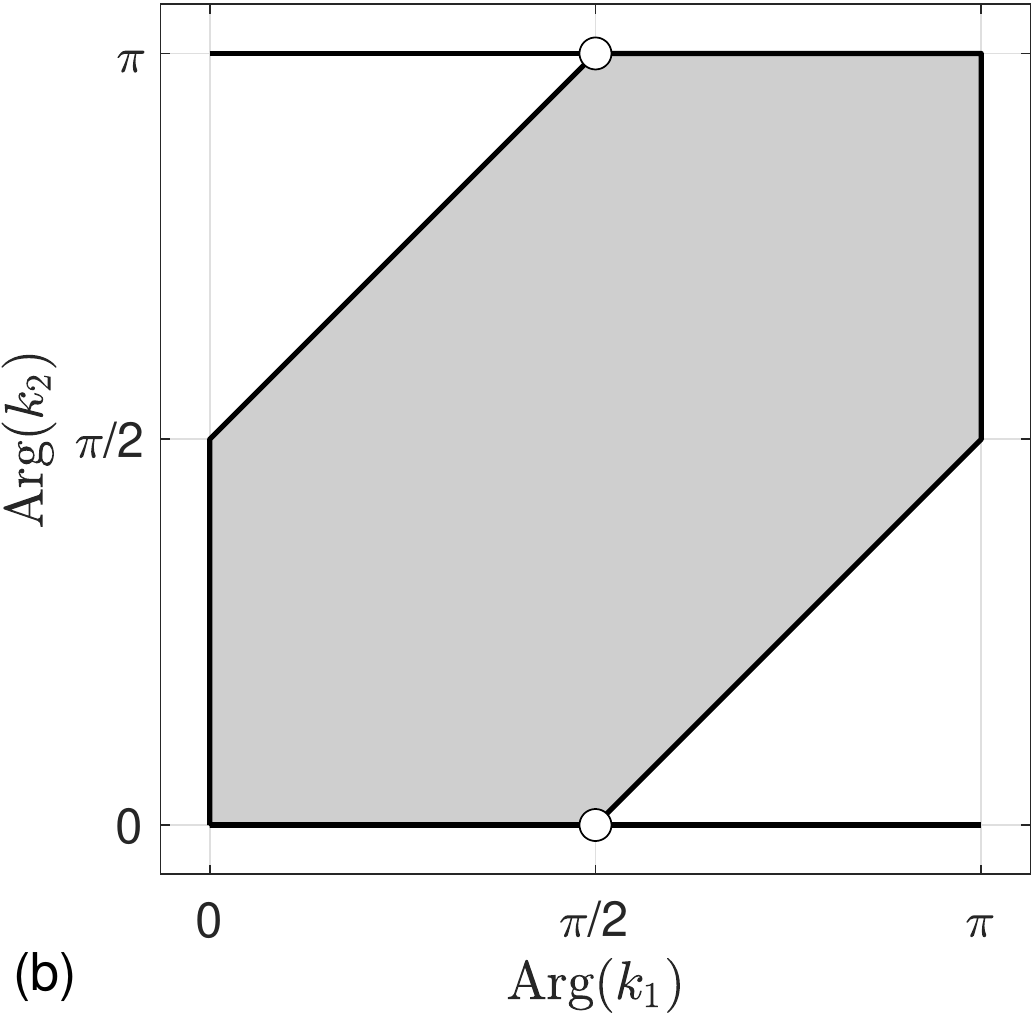}
\caption{\sf In each image, the gray region and the solid black lines 
  constitute a set of points $(\Arg(k_1),\Arg(k_2))$ for which, when
  $\kappa=k_2^2/k_1^2$, problem {\sf A} has at most one solution and
  problem ${\sf A}_0$ only has the trivial solution. Dashed lines and
  circles are not included: (a) a set of points obtained using
  techniques from \cite{KleiMart88,KresRoac78}; (b) the set of points
  discussed in the second paragraph of
  Section~\ref{sec:uniqueA0kappa}.}
\label{fig:hexagon}
\end{figure}

\subsection{Uniqueness and existence when $\kappa=k_2^2/k_1^2$}
\label{sec:uniqueA0kappa}

The parameter value $\kappa=k_2^2/k_1^2$ is relevant for the
electromagnetic transmission problem. By using similar techniques as
in \cite{KleiMart88,KresRoac78} one can show that when
$\kappa=k_2^2/k_1^2$ and $(\Arg(k_1),\Arg(k_2))$ is in the set of
points of Figure~\ref{fig:hexagon}(a), then problem {\sf A} has at
most one solution and problem ${\sf A}_0$ has only the trivial
solution $U(\myvec r)=0$.

We also mention that stronger results, including existence results,
are available for problem {\sf A} with~(\ref{eq:k1k2}) extended to
$0\leq\Arg(k_1),\Arg(k_2)\leq\pi$. Using methods from~\cite{Axels06},
developed for the more general Dirac equations, one can prove that
problem {\sf A}, with $\kappa=k_2^2/k_1^2$, has at most one solution
in finite energy norm for $(\Arg(k_1),\Arg(k_2))$ in the set of points
of Figure~\ref{fig:hexagon}(b). Furthermore, such a solution exists in
Lipschitz domains given that $\kappa\notin [-c_\Gamma,-1/c_\Gamma]$,
where $c_\Gamma\ge 1$ is a geometry-dependent constant which assumes
the value $c_\Gamma=1$ for smooth
$\Gamma$~\cite[Proposition~5.2]{HelsRose20}.

\section{Integral representations for problem {\sf A}}
\label{sec:intrepA}

We introduce two fields
\begin{align}
U_1(\myvec r)&= \frac{1}{2}K_{k_1}\mu(\myvec r)
               -\frac{1}{2}S_{k_1}\varrho(\myvec r)
               +U^{\rm in}(\myvec r)\,,\quad \myvec r\in\Omega_1\cup\Omega_2\,,
\label{eq:U1}\\
U_2(\myvec r)&=-\frac{1}{2}K_{k_2}\mu(\myvec r)
               +\frac{\kappa}{2}S_{k_2}\varrho(\myvec r)\,,
\quad\myvec r\in\Omega_1\cup\Omega_2\,,
\label{eq:U2}
\end{align}
where $\mu$ and $\varrho$ are unknown surface densities. The relations
in Section~\ref{sec:limits} give limits of $U_1(\myvec r)$ and
$U_2(\myvec r)$ at $\myvec r\in\Gamma$
\begin{align}
U_1^\pm(\myvec r)&=\pm\frac{1}{2}\mu(\myvec r)
             +\frac{1}{2}K_{k_1}\mu(\myvec r)
             -\frac{1}{2}S_{k_1}\varrho(\myvec r)
             +U^{\rm in}(\myvec r)\,,
\label{eq:U1pm}\\
U_2^\pm(\myvec r)&=\mp\frac{1}{2}\mu(\myvec r)
             -\frac{1}{2}K_{k_2}\mu(\myvec r)
             +\frac{\kappa}{2}S_{k_2}\varrho(\myvec r)\,.
\label{eq:U2pm}
\end{align}
Limits for the normal derivatives of $U_1(\myvec r)$ and $U_2(\myvec
r)$ at $\myvec r\in\Gamma$ are
\begin{align}
\myvec\nu\cdot[\nabla U_1]^\pm(\myvec r)&=\pm\frac{1}{2}\varrho(\myvec r)
             +\frac{1}{2}T_{k_1}\mu(\myvec r)
             -\frac{1}{2}K_{k_1}^{\rm A}\varrho(\myvec r)
             +\myvec\nu\cdot\nabla U^{\rm in}(\myvec r)\,,
\label{eq:gU1pm}\\
\myvec\nu\cdot[\nabla U_2]^\pm(\myvec r)&=\mp\frac{\kappa}{2}\varrho(\myvec r)
             -\frac{1}{2}T_{k_2}\mu(\myvec r)
             +\frac{\kappa}{2}K_{k_2}^{\rm A}\varrho(\myvec r)\,.
\label{eq:gU2pm}
\end{align}

We now construct the ansatz
\begin{equation}
U(\myvec r)=\left\{
\begin{array}{ll}
U_1(\myvec r)\,, & \myvec r\in\Omega_1\,,\\
U_2(\myvec r)\,, & \myvec r\in\Omega_2\,,
\end{array}
\right.
\label{eq:U}
\end{equation}
for the solution to problem {\sf A}. The fundamental solution
\eqref{eq:fund} makes $U$ of~(\ref{eq:U}) automatically satisfy the
PDEs of~(\ref{eq:helmA12}) and the radiation condition
\eqref{eq:BCA3}. It remains to determine $\mu$ and $\varrho$ to ensure
that the boundary conditions~(\ref{eq:BCA1}) and (\ref{eq:BCA2}) are
satisfied.

\section{Integral equations for problem {\sf A}}
\label{sec:inteqA}
 
We propose the system of second-kind integral equations on $\Gamma$
\begin{equation}
\begin{bmatrix}I-\beta_1(K_{k_1}-c_1K_{k_2})&\beta_1(S_{k_1}-c_1\kappa S_{k_2})\\
-\beta_2(T_{k_1}-c_2\kappa^{-1}T_{k_2})&I+\beta_2(K_{k_1}^{\rm A}-c_2K_{k_2}^{\rm A})
\end{bmatrix}
\begin{bmatrix}
\mu\\
\varrho
\end{bmatrix}=
2\begin{bmatrix}
\beta_1U^{\rm in}\\
\beta_2\partial_\nu U^{\rm in}
\end{bmatrix}
\label{eq:kombi}
\end{equation}
for the determination of $\mu$ and $\varrho$. Here $I$ is the identity and
\begin{equation}
\beta_i=(1+c_i)^{-1},\quad i=1,2\,,
\end{equation}
where $c_1$ and $c_2$ are two free parameters such that
\begin{equation}
c_i\neq -1,0\,,\quad i=1,2\,.
\label{eq:ci}
\end{equation}

We now prove that a solution $\{\mu,\varrho\}$ to~(\ref{eq:kombi}),
under certain conditions and via $U$ of~(\ref{eq:U}), represents a
solution to problem {\sf A}. Since $U$ of~(\ref{eq:U})
satisfies~(\ref{eq:helmA12}) and~(\ref{eq:BCA3}) for any
$\{\mu,\varrho\}$, it remains to show that $\{\mu,\varrho\}$
from~(\ref{eq:kombi}) makes $U$ satisfy~(\ref{eq:BCA1}) and
(\ref{eq:BCA2}). For this we need to show that, under certain
conditions, $U_1$ of (\ref{eq:U1}) is zero in $\Omega_2$ and $U_2$ of
(\ref{eq:U2}) is zero in $\Omega_1$. We introduce the auxiliary field
\begin{equation}
W(\myvec r)=\left\{
\begin{array}{ll}
         U_2(\myvec r)\,, & \myvec r\in\Omega_1\,,\\
-c_1^{-1}U_1(\myvec r)\,, & \myvec r\in\Omega_2\,.
\end{array}
\right.
\label{eq:W}
\end{equation}

The field $W$ of~(\ref{eq:W}), with $\{\mu,\varrho\}$
from~(\ref{eq:kombi}) and $U_1$ and $U_2$ from~(\ref{eq:U1})
and~(\ref{eq:U2}), is the unique solution to problem ${\sf B}_0$ with
$\alpha=c_2/(c_1\kappa)$ provided that the conditions of
Proposition~\ref{prop:uniqueB0} hold. This is so since $W$, by
construction, satisfies~(\ref{eq:helmB12}) and~(\ref{eq:BCB3}).
Furthermore, the boundary conditions~(\ref{eq:BCB1})
and~(\ref{eq:BCB2}) are satisfied. This can be checked by substituting
$U_1^-$ of \eqref{eq:U1pm} and $U_2^+$ of \eqref{eq:U2pm} into
\eqref{eq:BCB1}, and $\myvec\nu\cdot[\nabla U_1]^-$ of
\eqref{eq:gU1pm} and $\myvec\nu\cdot[\nabla U_2]^+$ of
\eqref{eq:gU2pm} into \eqref{eq:BCB2}, and using \eqref{eq:kombi}. As
a consequence, according to Proposition~\ref{prop:uniqueB0}, we have
\begin{equation}
W(\myvec r)=0\,,\quad\myvec r\in\Omega_1\cup\Omega_2\,.
\label{eq:W0}
\end{equation}

Several useful results for $\myvec r\in\Gamma$ follow
from~(\ref{eq:W}) and~(\ref{eq:W0})
\begin{align}
U_1^-(\myvec r)&=0\,,
\label{eq:01m}\\
U_2^+(\myvec r)&=0\,,
\label{eq:02p}\\
\myvec\nu\cdot[\nabla U_1]^-(\myvec r)&=0\,,
\label{eq:g01m}\\
\myvec\nu\cdot[\nabla U_2]^+(\myvec r)&=0\,.
\label{eq:g02p}
\end{align}
Now, from~(\ref{eq:U1pm}) and~(\ref{eq:01m}), and from~(\ref{eq:U2pm})
and~(\ref{eq:02p})
\begin{align}
U_1^+(\myvec r)&=\mu(\myvec r)\,,\quad\myvec r\in\Gamma\,,
\label{eq:mu1}\\
U_2^-(\myvec r)&=\mu(\myvec r)\,,\quad\myvec r\in\Gamma\,.
\label{eq:mu2}
\end{align}
Similarly, from~(\ref{eq:gU1pm}) and~(\ref{eq:g01m}), and
from~(\ref{eq:gU2pm}) and~(\ref{eq:g02p})
\begin{align}
\myvec\nu\cdot[\nabla U_1]^+(\myvec r)&=\varrho(\myvec r)\,,
\quad\myvec r\in\Gamma\,,
\label{eq:rho1}\\
\kappa^{-1}\myvec\nu\cdot[\nabla U_2]^-(\myvec r)&=
\varrho(\myvec r)\,,\quad\myvec r\in\Gamma\,.
\label{eq:rho2}
\end{align}
It is now easy to see that~(\ref{eq:BCA1}) and~(\ref{eq:BCA2}) are
satisfied and we conclude:
\begin{theorem}
  Assume that $\{k_1,k_2,\alpha=c_2/(c_1\kappa)\}$ is such that the
  conditions of Proposition~\ref{prop:uniqueB0} hold. Then a solution
  $\{\mu,\varrho\}$ to~(\ref{eq:kombi}) represents, via~(\ref{eq:U}),
  a solution also to problem {\sf A}. Furthermore, (\ref{eq:U})
  and~(\ref{eq:kombi}) correspond to a direct integral equation
  formulation of problem {\sf A} with $\mu$ and $\varrho$ linked to
  limits of $U$ and $\nabla U$ via~(\ref{eq:mu1})--(\ref{eq:rho2}).
\label{thm:exA}
\end{theorem}

\section{Unique solution to problem {\sf A} from~(\ref{eq:kombi})}
\label{sec:uniqueA}

We use the Fredholm alternative to prove that, under certain
conditions, the system~(\ref{eq:kombi}) has a unique solution
$\{\mu,\varrho\}$ and that this solution represents, via~(\ref{eq:U}),
the unique solution to problem {\sf A}. Three conditions are referred
to with roman numerals
\begin{itemize}
\item[(i)] $c_2=\kappa$ and~(\ref{eq:ci}) holds.
\item[(ii)] $k_1$, $k_2$, and $\kappa$ make the conditions of
  Proposition~\ref{prop:KR} or~\ref{prop:KM} hold or, if
  $\kappa=k_2^2/k_1^2$, $(\Arg(k_1),\Arg(k_2))$ is in the set of
  points of Figure~\ref{fig:hexagon}(a).
\item[(iii)] $\{k_1,k_2,\alpha=c_2/(c_1\kappa)\}$ makes the conditions
  of Proposition~\ref{prop:uniqueB0} hold.
\end{itemize}

We start with the observation that (\ref{eq:kombi}) is a Fredholm
second-kind integral equation with compact (differences of) operators
when condition (i) holds and $\Gamma$ is smooth. Then the Fredholm
alternative can be applied to~(\ref{eq:kombi}). Let $\mu_0$ and
$\varrho_0$ be solutions to the homogeneous version
of~(\ref{eq:kombi}). Let $U_{10}$, $U_0$, and $W_0$ be the
fields~(\ref{eq:U1}), (\ref{eq:U}), and~(\ref{eq:W}) with $\mu=\mu_0$
and $\varrho=\varrho_0$. From Section~\ref{sec:inteqA} we know that
$W_0=0$ if (iii) holds. We shall now prove that also $U_0=0$ and, from
that, $\mu_0=0$ and $\varrho_0=0$.

It follows from Theorem~\ref{thm:exA}, which requires (iii), that
$\{\mu_0,\varrho_0\}$ represents a solution to problem ${\sf A}_0$.
If (ii) holds, then $U_0=0$ according to Corollary~\ref{cor:uniqueA0}.
It then follows that $U_{10}=0$ in $\Omega_1$ so that $U_{10}^+=0$ and
$[\nabla U_{10}]^+=0$. Then $\mu_0=0$ and $\varrho_0=0$
from~(\ref{eq:mu1}) and~(\ref{eq:rho1}). Now, from the Fredholm
alternative, the system (\ref{eq:kombi}) has a unique solution
$\{\mu,\varrho\}$. By Theorem~\ref{thm:exA} this solution represents a
solution to problem {\sf A}. If problem {\sf A} has at most one
solution, which requires (ii), this solution to problem {\sf A} is
unique and we conclude:
\begin{theorem}
  Assume that conditions (i), (ii), (iii) hold. Then the
  system~(\ref{eq:kombi}) has a unique solution $\{\mu,\varrho\}$
  which represents the unique solution to problem {\sf A}.
\label{thm:exunA}
\end{theorem} 

Note that, when (i) holds, $\alpha=1/c_1$ in (iii) and it is always
possible to find a constant $c_1$ so that~(\ref{eq:crit2KLB0}) holds
under the assumption~(\ref{eq:k1k2}). In this respect, condition (iii)
in Theorem~\ref{thm:exunA} does not introduce any additional
constraint to problem {\sf A}. A simple rule that satisfies condition
(iii) is 
\begin{equation} 
c_1=\left\{ 
\begin{array}{lll} 
e^{{\rm i}\Arg(k_2)} &\mbox{if}&\Re{\rm e}\{k_1\}\ge 0\,,\\
e^{{\rm i}(\Arg(k_2)-\pi)}&\mbox{if}&\Re{\rm e}\{k_1\}< 0\,.
\end{array}\right.
\end{equation} 
This rule gives $c_1={\rm i}$ when $(\Arg(k_1),\Arg(k_2))=(0,\pi/2)$.
It is also possible to choose $c_1=-{\rm i}$ when
$(\Arg(k_1),\Arg(k_2))=(0,\pi/2)$.

Our results, so far, extend those of~\cite[Section~4.1]{KleiMart88},
where a direct formulation of problem {\sf A} is presented
in~\cite[Eq.~(4.10)]{KleiMart88}. To see this, note that
\cite[Eq.~(4.10)]{KleiMart88} corresponds to \eqref{eq:kombi} with
$c_2=\kappa$ and $c_1=1/\kappa$. Now~(\ref{eq:kombi}) with
$c_2=\kappa$ and $c_1$ in agreement with~(\ref{eq:crit2KLB0}) provides
unique solutions over a broader range of $k_1$, $k_2$, and $\kappa$
than does~\cite[Eq.~(4.10)]{KleiMart88}. For example, if
$(\Arg(k_1),\Arg(k_2))=(0,\pi/2)$ and $\Arg(\kappa)=\pi$, then
\eqref{eq:kombi} with $c_2=\kappa$ and $c_1=\pm{\rm i}$ is guaranteed
to have a unique solution while~\cite[Eq.~(4.10)]{KleiMart88} is not.

\section{A weakly singular representation of $U$}
\label{sec:eval}

Once the solution $\{\mu,\varrho\}$ has been obtained
from~(\ref{eq:kombi}), the field $U(\myvec r)$ can be evaluated
via~(\ref{eq:U}). When $\myvec r$ is close to $\Gamma$, this could be
problematic due to the rapid variation with $\myvec r'$ in the
Cauchy-type singular kernels of $K_{k_1}$ and $K_{k_2}$
in~(\ref{eq:U1}) and~(\ref{eq:U2}). To alleviate this problem we
introduce
\begin{equation}
V(\myvec r)=\left\{
\begin{array}{ll}
U_2(\myvec r)\,, & \myvec r\in\Omega_1\,,\\
U_1(\myvec r)\,, & \myvec r\in\Omega_2\,.
\end{array}
\right.
\label{eq:V}
\end{equation}
From~(\ref{eq:W}) and~(\ref{eq:W0}) it follows that $V$ is a
null-field such that $V=0$ in $\Omega_1\cup\Omega_2$, and hence
$U=U+V$. The Cauchy-type kernel singularities in the representation of
$U+V$ cancel out and we are left with better-behaved weakly singular
kernels. In the numerical examples in Section~\ref{sec:numex} we
exploit $U=U+V$ for near-field evaluation.

\section{Electromagnetic transmission problems}
\label{sec:EMprob}

We present three electromagnetic transmission problems called problem
{\sf C}, problem ${\sf C}_0$, and problem ${\sf D}_0$. The main
problem is {\sf C}, whereas problems ${\sf C}_0$ and ${\sf D}_0$ are
needed in proofs.

The prerequisites in Section \ref{sec:pre} hold, with regions
$\Omega_1$ and $\Omega_2$ that are dielectric and non-magnetic. The
electric field is denoted $\myvec E$ and the magnetic field $\myvec
H$. The electric field is scaled such that $\myvec E=\eta_1^{-1}\myvec
E_{\rm unscaled}$, where $\eta_1=\sqrt{\mu_0/\varepsilon_1}$ is the
wave impedance of $\Omega_1$ and $\varepsilon_1$ is the permittivity
of $\Omega_1$. Furthermore, problems {\sf C}, ${\sf C}_0$, and ${\sf
  D}_0$ contain a complex parameter $\kappa$ which plays a somewhat
similar role as the parameter $\kappa$ of Section \ref{sec:A0} played
in problem {\sf A} and ${\sf A}_0$. This new $\kappa$ has the value
$\kappa=\varepsilon_2/\varepsilon_1$, where $\varepsilon_2$ is the
permittivity of $\Omega_2$. For non-magnetic materials, this is
equivalent to
\begin{equation}
\kappa=k_2^2/k_1^2.
\label{eq:kappa}
\end{equation}

\subsection{Problems {\sf C} and  ${\sf C}_0$}
\label{sec:C0}

The transmission problem {\sf C} reads: Given an incident field
$\myvec H^{\rm in}$, generated in $\Omega_1$, find $\myvec E(\myvec
r),\,\myvec H(\myvec r)$, $\myvec r\in\Omega_1\cup\Omega_2$, which,
for wavenumbers $k_1$ and $k_2$ and with $\kappa$
from~(\ref{eq:kappa}) such that
\begin{equation}
0\leq\Arg(k_1), \Arg(k_2)<\pi
\quad\mbox{and}\quad
\kappa\neq -1\,,
\label{eq:k1k2Maxkap}
\end{equation}
solve Maxwell's equations
\begin{equation}
\begin{split}
\nabla\times\myvec E(\myvec r)&= {\rm i}k_1\myvec H(\myvec r)\,,
\quad\myvec r\in\Omega_1\cup\Omega_2\,,\\
\nabla\times\myvec H(\myvec r)&=-{\rm i}k_1\myvec E(\myvec r)\,,
\quad\myvec r\in\Omega_1\,,\\
\nabla\times\myvec H(\myvec r)&=-{\rm i}k_1\kappa\myvec E(\myvec r)\,,
\quad\myvec r\in\Omega_2\,,
\end{split}
\label{eq:Max123C}
\end{equation}
except possibly at an isolated point in $\Omega_1$ where the source of
$\myvec H^{\rm in}$ is located, subject to the boundary conditions
\begin{align}
\myvec\nu\times\myvec E^+(\myvec r)&=\myvec\nu\times\myvec E^-(\myvec r)\,,
\quad\myvec r\in \Gamma\,,\label{eq:rv2C}\\
\myvec\nu\times\myvec H^+(\myvec r)&=\myvec\nu\times\myvec H^-(\myvec r)\,, 
\quad\myvec r\in \Gamma\,,\label{eq:rv1C}\\
\left(\partial_{\hat{\myvec r}}-{\rm i}k_1\right)\myvec H^{\rm sc} 
(\myvec r)&=o\left(\vert \myvec r\vert^{-1}\right),
\quad\lvert\myvec r\rvert\rightarrow\infty\,.
\label{eq:radcondC}
\end{align}
The scattered field $\myvec H^{\rm sc}$ is source free in $\Omega_1$
and defined by
\begin{equation}
\myvec H(\myvec r)=\myvec H^{\rm in}(\myvec r)+\myvec H^{\rm sc}(\myvec r)\,,
\quad\myvec r\in \Omega_1\,.
\label{eq:decomp}
\end{equation}
The condition \eqref{eq:radcondC} and decomposition \eqref{eq:decomp}
also hold for $\myvec E$. The incident field satisfies
\begin{equation}
\begin{split}
\nabla\times\myvec E^{\rm in}(\myvec r)&= {\rm i}k_1\myvec H^{\rm in}
(\myvec r)\,,\quad\myvec r\in\mathbb{R}^3\,,\\
\nabla\times\myvec H^{\rm in}(\myvec r)&=-{\rm i}k_1\myvec E^{\rm in}
(\myvec r)\,,\quad\myvec r\in\mathbb{R}^3\,,
\end{split}
\end{equation}
except at the possible isolated source point in $\Omega_1$.

The homogeneous problem ${\sf C}_0$ is problem {\sf C} with $\myvec
E^{\rm in}=\myvec H^{\rm in}=\myvec 0$.

\subsection{Problem ${\sf D}_0$}
\label{sec:D0}

The transmission problem ${\sf D}_0$ reads: find $\myvec E_W(\myvec
r),\,\myvec H_W(\myvec r)$, $\myvec r\in\Omega_1\cup\Omega_2$, which,
for a complex jump parameter $\lambda$ and for $k_1$, $k_2$, and
$\kappa$ such that \eqref{eq:k1k2Maxkap} holds, solve
\begin{equation}
\begin{split}
\nabla\times\myvec E_W(\myvec r)&= {\rm i}k_1\myvec H_W(\myvec r)\,,
\quad\myvec r\in\Omega_1\cup\Omega_2\,,\\
\nabla\times\myvec H_W(\myvec r)&=-{\rm i}k_1\kappa\myvec E_W(\myvec r)\,,
\quad\myvec r\in\Omega_1\,,\\
\nabla\times\myvec H_W(\myvec r)&=-{\rm i}k_1\myvec E_W(\myvec r)\,,
\quad\myvec r\in\Omega_2\,,
\end{split}
\label{eq:D0123}
\end{equation}
subject to the boundary conditions
\begin{align}
\lambda\kappa\myvec\nu\times \myvec E_W^+(\myvec r)&=
\myvec\nu\times\myvec E_W^-(\myvec r)\,,\quad\myvec r\in\Gamma\,,
\label{eq:RV2D0}\\
\myvec\nu\times\myvec H_W^+(\myvec r)&=
\myvec\nu\times\myvec H_W^-(\myvec r)\,,\quad\myvec r\in\Gamma\,,
\label{eq:RV1D0}\\
\left(\partial_{\hat{\myvec r}}-{\rm i}k_2\right)\myvec H_W(\myvec r)&=o
\left(\vert\myvec r\vert^{-1}\right),
\quad\lvert\myvec r\rvert\rightarrow\infty\,.
\label{eq:radcondD0}
\end{align}
The radiation condition \eqref{eq:radcondD0} also holds for $\myvec
E_W$.

\subsection{Uniqueness and existence of solutions to problem ${\sf
     C}$, ${\sf C}_0$, and ${\sf D}_0$}
\label{sec:uniqueex}

In Appendix D it is shown that when $(\Arg(k_1),\Arg(k_2))$ is in the
set of points of Figure~\ref{fig:hexagon}(a), then problem {\sf C} has
at most one solution and problem ${\sf C}_0$ has only the trivial
solution $\myvec E=\myvec H=\myvec 0$. It is also shown that when the
conditions of Proposition~\ref{prop:uniqueB0} hold for
$\{k_1,k_2,\alpha=\lambda\}$, then problem ${\sf D}_0$ has only the
trivial solution $\myvec E_W= \myvec H_W=\myvec 0$.

The stronger results for problem {\sf A}, discussed in
Section~\ref{sec:uniqueA0kappa}, carry over to problem {\sf C}. One
can prove that there exists a unique solution in finite energy norm to
problem {\sf C} in Lipschitz domains when $(\Arg(k_1),\Arg(k_2))$ is
in the set of points of Figure~\ref{fig:hexagon}(b) and $\kappa$ is
outside a certain interval on the real
axis~\cite[Proposition~7.4]{HelsRose20}.

\section{Integral representations for problem {\sf C}}
\label{sec:intrepC}

Let $\sigma_{\rm E}$, $\varrho_{\rm E}$, $\myvec M_{\rm s}$, $\myvec
J_{\rm s}$, $\varrho_{\rm M}$, and $\sigma_{\rm M}$ be six unknown,
scalar- and vector-valued, surface densities and define the four
fields
\begin{align}
\myvec E_1(\myvec r)&=
-\frac{1}{2}\boldsymbol{\cal N}_{k_1}\varrho_{\rm E}(\myvec r)
-\frac{1}{2}{\cal R}_{k_1}\left(\myvec\nu'\sigma_{\rm M}
+\myvec M_{\rm s})(\myvec r\right)\nonumber\\
&\qquad
+\frac{1}{2}\tilde{\cal S}_{k_1}(\myvec\nu'\sigma_{\rm E}
+\myvec J_{\rm s})(\myvec r)+ \myvec E^{\rm in}(\myvec r)\,,
\quad\myvec r\in\Omega_1\cup\Omega_2\,,
\label{eq:E1rep1}\\
\myvec E_2(\myvec r)&=
\frac{1}{2\kappa}\boldsymbol{\cal N}_{k_2}\varrho_{\rm E}(\myvec r)
+\frac{1}{2\kappa}{\cal R}_{k_2}\left(\myvec\nu'\sigma_{\rm M}
+\kappa\myvec M_{\rm s})(\myvec r\right)\nonumber\\
&\qquad 
-\frac{1}{2}\tilde{\cal S}_{k_2}(\kappa^{-1}\myvec\nu'\sigma_{\rm E}
+\myvec J_{\rm s})(\myvec r)\,,
\quad\myvec r\in\Omega_1\cup\Omega_2\,,
\label{eq:E2rep1}
\end{align}
\begin{align}
\myvec H_1(\myvec r)&=
\frac{1}{2}\tilde{\cal S}_{k_1}\left(\myvec\nu'\sigma_{\rm M}
+\myvec M_{\rm s})(\myvec r\right)
+\frac{1}{2}{\cal R}_{k_1}(\myvec\nu'\sigma_{\rm E}
+\myvec J_{\rm s})(\myvec r)\nonumber\\
&\qquad 
-\frac{1}{2}\boldsymbol{\cal N}_{k_1}\varrho_{\rm M}(\myvec r)
+\myvec H^{\rm in}(\myvec r)\,,
\quad\myvec r\in \Omega_1\cup \Omega_2\,,
\label{eq:H1rep11}\\
\myvec H_2(\myvec r)&=
-\frac{1}{2}\tilde{\cal S}_{k_2}\left(\myvec\nu'\sigma_{\rm M}
+\kappa\myvec M_{\rm s})(\myvec r\right)
-\frac{1}{2}{\cal R}_{k_2}(\kappa^{-1}\myvec\nu'\sigma_{\rm E}
+\myvec J_{\rm s})(\myvec r)\nonumber\\
&\qquad 
+\frac{1}{2}\boldsymbol{\cal N}_{k_2}\varrho_{\rm M}(\myvec r)\,,
\quad\myvec r\in\Omega_1\cup\Omega_2\,.
\label{eq:H2rep11}
\end{align}
The introduction of $\sigma_{\rm E}$ and $\sigma_{\rm M}$ is inspired
by the integral representations for the generalized Helmholtz
transmission problem in \cite{Vico16,VicGreFer18}.

The integral representations of the fields $\myvec E$ and $\myvec H$
for problem {\sf C} are
\begin{equation}
\myvec E(\myvec r)=\left\{
\begin{array}{ll}
\myvec E_1(\myvec r)\,, & \myvec r\in\Omega_1\,,\\
\myvec E_2(\myvec r)\,, & \myvec r\in\Omega_2\,,
\end{array}\right.
\quad
\myvec H(\myvec r)=\left\{
\begin{array}{ll}
\myvec H_1(\myvec r)\,, & \myvec r\in \Omega_1\,,\\
\myvec H_2(\myvec r)\,, & \myvec r\in \Omega_2\,.
\end{array}\right.
\label{eq:repEH}
\end{equation}

\section{Integral equations for problem {\sf C}}
\label{sec:inteqC}

For the determination of $\{\sigma_{\rm E},\varrho_{\rm E},\myvec
M_{\rm s},\myvec J_{\rm s},\varrho_{\rm M},\sigma_{\rm M}\}$ we
propose the system of second-kind integral equations on $\Gamma$
\begin{equation}
\left(I+{\bf D}{\bf Q}\right)\myvec\mu=2{\bf D}\myvec f\,.
\label{eq:EHsys}
\end{equation}
Here $\myvec\mu$ and $\myvec f$ are column vectors with six entries
each
\begin{align*}
\myvec\mu&=\left[\sigma_{\rm E};\varrho_{\rm E};\myvec M_{\rm s};
\myvec J_{\rm s};\varrho_{\rm M};\sigma_{\rm M}\right],\\
\myvec f&=\left[0;\myvec\nu\cdot\myvec E^{\rm in};
-\myvec\nu\times\myvec E^{\rm in};\myvec\nu\times\myvec H^{\rm in};
\myvec\nu\cdot\myvec H^{\rm in};0\right],
\end{align*}
${\bf Q}$ is a $6\times 6$ matrix whose non-zero operator entries
$Q_{ij}$ map scalar- or vector-valued densities to scalar or
vector-valued functions
\begin{align*}
Q_{11}&=-K_{k_1}+c_3K_{k_2}\,,\quad
Q_{12} =-\tilde S_{k_1}+c_3\kappa\tilde S_{k_2}\,,\quad
Q_{14} =\nabla\cdot({\cal S}_{k_1}-c_3\kappa {\cal S}_{k_2})\,,\\
Q_{21}&=-\myvec\nu\cdot(\tilde{\cal S}_{k_1}-c_4\tilde{\cal
  S}_{k_2})\myvec\nu'\,,\quad
Q_{22} =K_{k_1}^ {\rm A}-c_4K_{k_2}^{\rm A}\,,\\
Q_{23}&=\myvec\nu\cdot({\cal R}_{k_1}-c_4\kappa {\cal R}_{k_2})\,,\quad
Q_{24} =-\myvec\nu\cdot(\tilde{\cal S}_{k_1}-c_4\kappa\tilde{\cal S}_{k_2})\,,\\
Q_{26}&=\myvec\nu\cdot({\cal R}_{k_1}-c_4{\cal R}_{k_2})\myvec\nu'\,,\quad
Q_{31} =\myvec\nu\times(\tilde{\cal S}_{k_1}-c_5\kappa^{-1}\tilde{\cal S}_{k_2})\myvec\nu'\,,\\
Q_{32}&=-\myvec\nu\times(\boldsymbol{\cal
  N}_{k_1}-c_5\kappa^{-1}\boldsymbol{\cal N}_{k_2})\,,\quad
Q_{33} =-\myvec\nu\times({\cal R}_{k_1}-c_5{\cal R}_{k_2})\,,\\
Q_{34}&=\myvec\nu\times(\tilde{\cal S}_{k_1}-c_5\tilde{\cal S}_{k_2})\,,\quad
Q_{36} =-\myvec\nu\times({\cal R}_{k_1}-c_5\kappa^{-1}{\cal
  R}_{k_2})\myvec\nu'\,,\\
Q_{41}&=-\myvec\nu\times({\cal R}_{k_1}-c_6\kappa^{-1}{\cal
  R}_{k_2})\myvec\nu'\,,\quad
Q_{43} =-\myvec\nu\times(\tilde{\cal S}_{k_1}-c_6\kappa\tilde{\cal
  S}_{k_2})\,,\\
Q_{44}&=-\myvec\nu\times({\cal R}_{k_1}-c_6{\cal R}_{k_2})\,,\quad
Q_{45} =\myvec\nu\times(\boldsymbol{\cal N}_{k_1}-c_6\boldsymbol{\cal
  N}_{k_2})\,,\\
Q_{46}&=-\myvec\nu\times(\tilde{\cal S}_{k_1}-c_6\tilde{\cal
  S}_{k_2})\myvec\nu'\,,\quad
Q_{51} =-\myvec\nu\cdot({\cal R}_{k_1}-c_7\kappa^{-1}{\cal
  R}_{k_2})\myvec\nu'\,,\\
Q_{53}&=-\myvec\nu\cdot(\tilde{\cal S}_{k_1}-c_7\kappa \tilde{\cal
  S}_{k_2})\,,\quad
Q_{54} =-\myvec\nu\cdot({\cal R}_{k_1}-c_7{\cal R}_{k_2})\,,\\
Q_{55}&=K_{k_1}^{\rm A}-c_7K_{k_2}^{\rm A}\,,\quad
Q_{56} =-\myvec\nu\cdot(\tilde{\cal S}_{k_1}-c_7\tilde{\cal
  S}_{k_2})\myvec\nu'\,,\\
Q_{63}&=\nabla\cdot({\cal S}_{k_1}-c_8\kappa{\cal S}_{k_2})\,,\quad
Q_{65} =-\tilde S_{k_1}+c_8\kappa\tilde S_{k_2}\,,\quad
Q_{66} =-K_{k_1}+c_8K_{k_2}\,,
\end{align*}
${\bf D}$ is a diagonal $6\times 6$ matrix of scalars with non-zero
entries
\begin{equation}
\begin{split}
&D_{ii}=(1+c_{i+2})^{-1}\,,\quad i=1,\ldots,6\,,\\
c_3=&\gamma_{\rm E}c\,,\quad c_4=c_6=c\,,\quad c_5=c_7=\lambda\kappa c\,,
\quad c_8=\gamma_{\rm M}c\,,
\end{split}
\end{equation}
and $c$, $\lambda$, $\gamma_{\rm E}$, and $\gamma_{\rm M}$ are free
parameters such that
\begin{equation}
c_i \neq -1,0,\quad i=3,\ldots,8\,.
\label{eq:neq0}
\end{equation}

\subsection{Criteria for \eqref{eq:repEH} to represent a solution 
  to problem {\sf C}}

We now prove that a solution $\myvec\mu$ to \eqref{eq:EHsys}, under
certain conditions and via \eqref{eq:repEH}, represents a solution to
problem {\sf C}.
   
The fundamental solution \eqref{eq:fund} makes $\myvec E$ and $\myvec
H$ of \eqref{eq:repEH} satisfy the radiation condition
\eqref{eq:radcondC}. It remains to prove that $\myvec E$ and $\myvec
H$ satisfy Maxwell's equations \eqref{eq:Max123C} and the boundary
conditions \eqref{eq:rv2C} and \eqref{eq:rv1C}. For this we first need
to show that, under certain conditions, $\myvec E_1$ and $\myvec H_1$
of \eqref{eq:E1rep1} and \eqref{eq:H1rep11} are zero in $\Omega_2$ and
$\myvec E_2$ and $\myvec H_2$ of \eqref{eq:E2rep1} and
\eqref{eq:H2rep11} are zero in $\Omega_1$. We introduce the auxiliary
fields
\begin{equation}
\myvec E_W(\myvec r)=\left\{
\begin{array}{ll}
\myvec E_2(\myvec r)\,,        & \myvec r\in\Omega_1\,,\\
-c^{-1}\myvec E_1(\myvec r)\,, & \myvec r\in\Omega_2\,,
\end{array}\right.
\quad
\myvec H_W(\myvec r)=\left\{
\begin{array}{ll}
\myvec H_2(\myvec r)\,,        & \myvec r\in\Omega_1\,,\\
-c^{-1}\myvec H_1(\myvec r)\,, & \myvec r\in\Omega_2\,.
\end{array}\right.
\label{eq:EHW1}
\end{equation}

The fields $\myvec E_W$ and $\myvec H_W$, with $\myvec\mu$ from
\eqref{eq:EHsys}, is the unique trivial solution to problem ${\sf
  D}_0$ provided the sets $\{k_1,k_2,\alpha=\lambda\bar\gamma_{\rm
  M}\}$, $\{k_1,k_2,\alpha=\bar\gamma_{\rm E}\bar\kappa\}$, and
$\{k_1,k_2,\alpha=\lambda\}$ are such that the conditions of
Proposition~\ref{prop:uniqueB0} hold. This statement is now shown in
several steps. The fundamental solution \eqref{eq:fund} makes $\myvec
E_W$ and $\myvec H_W$ satisfy \eqref{eq:radcondD0}. Using Appendix A
in combination with \eqref{eq:EHsys} one can show that
\eqref{eq:RV2D0} and \eqref{eq:RV1D0} are satisfied. Appendix B shows
that if $\myvec\mu$ is a solution to \eqref{eq:EHsys} and if the
conditions of Proposition~\ref{prop:uniqueB0} hold for
$\{k_1,k_2,\alpha=\lambda\bar\gamma_{\rm M}\}$ and
$\{k_1,k_2,\alpha=\bar\gamma_{\rm E}\bar\kappa\}$, then
\begin{align}
\nabla\cdot\myvec E_W(\myvec r)&=0\,,
\quad\myvec r\in \Omega_1\cup\Omega_2\,,\label{eq:divEW}\\
\nabla\cdot\myvec H_W(\myvec r)&=0\,,
\quad\myvec r\in \Omega_1\cup\Omega_2\,.\label{eq:divHW}
\end{align}
Appendix C shows that if \eqref{eq:divEW} and \eqref{eq:divHW} hold,
then $\myvec E_W$ and $\myvec H_W$ satisfy \eqref{eq:D0123}. If the
conditions of Proposition~\ref{prop:uniqueB0} also hold for
$\{k_1,k_2,\alpha=\lambda\}$, then ${\sf D}_0$ only has the trivial
solution $\myvec E_W=\myvec H_W=\myvec 0$, that is,
\begin{equation}
\begin{split}
\myvec E_2(\myvec r)&=\myvec H_2(\myvec r)=
\myvec 0\,,\quad\myvec r\in \Omega_1\,,\\
\myvec E_1(\myvec r)&=\myvec H_1(\myvec r)=
\myvec 0\,,\quad\myvec r\in \Omega_2\,.
\end{split}
\label{eq:H1H2zero}
\end{equation}
By that the statement is proven.

From \eqref{eq:H1H2zero} and Appendix A we obtain the surface
densities as boundary values of the full 3D electromagnetic fields
\begin{align}
[\nabla\cdot \myvec E_1]^+(\myvec r)&=
\kappa[\nabla\cdot \myvec E_2]^-(\myvec r)=
-{\rm i}k_1\sigma_{\rm E}(\myvec r)\,,\label{eq:bvalues1}\\
\myvec\nu\cdot\myvec E_1^+(\myvec r)&=
\kappa\myvec\nu\cdot\myvec E_2^-(\myvec r)=
\varrho_{\rm E}(\myvec r)\,,\label{eq:bvalues2}\\
\myvec\nu\times\myvec E_1^+(\myvec r)&=
\myvec\nu\times\myvec E_2^-(\myvec r)=
-\myvec M_{\rm s}(\myvec r)\,,\label{eq:bvalues3}\\
\myvec\nu\times\myvec H_1^+(\myvec r)&=
\myvec\nu\times\myvec H_2^-(\myvec r)=
\myvec J_{\rm s}(\myvec r)\,,\label{eq:bvalues4}\\
\myvec\nu\cdot\myvec H_1^+(\myvec r)&=
\myvec\nu\cdot\myvec H_2^-(\myvec r)=
\varrho_{\rm M}(\myvec r)\,,\label{eq:bvalues5}\\
[\nabla\cdot \myvec H_1]^+(\myvec r)&=
[\nabla\cdot \myvec H_2]^-(\myvec r)=-{\rm i}k_1\sigma_{\rm M}(\myvec r)\,.
\label{eq:bvalues6}
\end{align}
Due to \eqref{eq:bvalues3} and \eqref{eq:bvalues4}, $\myvec E$ and
$\myvec H$ of \eqref{eq:repEH} satisfy \eqref{eq:rv2C} and
\eqref{eq:rv1C}. Appendix B shows that
\eqref{eq:bvalues1}--\eqref{eq:bvalues6} imply
\begin{equation}
\nabla\cdot\myvec E(\myvec r)=\nabla\cdot\myvec H(\myvec r)=0\,,
\quad\myvec r\in\Omega_1\cup\Omega_2\,,
\label{eq:divH0E0}
\end{equation}
when $(\Arg(k_1),\Arg(k_2))$ is in the set of points of Figure
\ref{fig:hexagon}(a). Finally, from the representations
\eqref{eq:E1rep1}--\eqref{eq:H2rep11} and the divergence condition
\eqref{eq:divH0E0}, Appendix C shows that \eqref{eq:Max123C} is
satisfied. We conclude:
\begin{theorem}
  Assume that $\{k_1,k_2,\alpha=\lambda\bar\gamma_{\rm M}\}$,
  $\{k_1,k_2,\alpha=\bar\gamma_{\rm E}\bar\kappa\}$, and $\{k_1,k_2,
  \alpha=\lambda\}$ are such that the conditions of
  Proposition~\ref{prop:uniqueB0} hold. Then a solution $\myvec\mu$ to
  \eqref{eq:EHsys} represents, via \eqref{eq:repEH}, a solution also
  to problem {\sf C}. Furthermore, \eqref{eq:repEH} and
  \eqref{eq:EHsys} correspond to a direct integral equation
  formulation of problem {\sf C} with $\myvec\mu$ linked to limits of
  $\myvec E$ and $\myvec H$ via
  \eqref{eq:bvalues1}--\eqref{eq:bvalues6}.
\label{thm:exC}
\end{theorem}

\begin{remark}
  The surface densities in \eqref{eq:bvalues1}--\eqref{eq:bvalues6}
  can be given the following physical interpretations: $-{\rm
    i}k_1\sigma_{\rm E}$ and $-{\rm i}k_1\sigma_{\rm M}$ are the
  electric and magnetic volume charge densities at $\Gamma^+$,
  $\varrho_{\rm E}$ and $\varrho_{\rm M}$ are the equivalent electric
  and magnetic surface charge densities on $\Gamma^+$, and $\myvec
  M_{\rm s}$ and $\myvec J_{\rm s}$ are the equivalent magnetic and
  electric surface current densities on $\Gamma^+$.
\end{remark}

\section{Unique solution to  problem {\sf C} from \eqref{eq:EHsys}}
\label{sec:uniqueC}

We now prove that if there exists a solution to problem {\sf C}, then,
under certain conditions, there exists a solution $\myvec\mu$ to
\eqref{eq:EHsys} and it represents the unique solution to problem {\sf
  C}. Three conditions are referred to
\begin{itemize}
\item[(i)] The conditions in \eqref{eq:neq0} hold.
\item[(ii)] $(\Arg(k_1),\Arg(k_2))$ is in the set of points of
  Figure~\ref{fig:hexagon}(a).
\item[(iii)] $\{k_1,k_2,\alpha=\lambda\bar\gamma_{\rm M}\}$,
  $\{k_1,k_2,\alpha=\bar\gamma_{\rm E}\bar\kappa\}$, and $\{k_1, k_2,
  \alpha=\lambda\}$ are such that the conditions of
  Proposition~\ref{prop:uniqueB0} hold.
\end{itemize}
Let $\myvec\mu_{0}$ be a solution to the homogeneous version of
\eqref{eq:EHsys} and assume that (i), (ii), and (iii) hold. Since
(iii) holds, $\myvec\mu_{\rm 0}$ represents a solution to problem {\sf
  C}$_0$, according to Theorem \ref{thm:exC}. Since (ii) holds, this
solution is the trivial solution $\myvec E=\myvec H=\myvec 0$,
according to Section \ref{sec:uniqueex}. The limits of fields in
\eqref{eq:bvalues1}--\eqref{eq:bvalues6} are then zero and hence
$\myvec\mu_0=\myvec 0$. Then \eqref{eq:EHsys} has at most one solution
$\myvec\mu$. Since $\myvec\mu$ is linked to limits of $\myvec E$ and
$\myvec H$ via \eqref{eq:bvalues1}--\eqref{eq:bvalues6} it follows
that if problem {\sf C} has a solution, then via
\eqref{eq:bvalues1}--\eqref{eq:bvalues6} this solution gives
a $\myvec\mu$ that solves \eqref{eq:EHsys}. We conclude:
\begin{theorem}
  Assume that there exists a solution to problem {\sf C} and that
  condition (ii) holds. Then this solution is unique. If conditions
  (i) and (iii) also hold, then there exists a unique solution
  $\myvec\mu$ to \eqref{eq:EHsys} and this solution represents via
  \eqref{eq:repEH} the unique solution to problem {\sf C}.
\label{thm:exunC}
\end{theorem}
\begin{remark}
  From \eqref{eq:bvalues1}, \eqref{eq:bvalues6}, and \eqref
  {eq:divH0E0} it is seen that $\sigma_{\rm E}$ and $\sigma_{\rm M}$
  are zero. Despite this, $\sigma_{\rm E}$ and $\sigma_{\rm M}$ are
  needed in \eqref{eq:EHsys} to guarantee uniqueness. Often, however,
  one can omit $\sigma_{\rm E}$ and $\sigma_{\rm M}$ from
  \eqref{eq:EHsys} and still get the correct unique solution.
\label{rem:div}
\end{remark}

\subsection{Determination of uniqueness parameters}

The system \eqref{eq:EHsys} contains the free parameters $\lambda$,
$\gamma_{\rm E}$, $\gamma_{\rm M}$, and $c$. Unique solvability of
\eqref{eq:EHsys} requires that the conditions of
Proposition~\ref{prop:uniqueB0} hold for the sets
$\{k_1,k_2,\alpha=\lambda\bar\gamma_{\rm M}\}$,
$\{k_1,k_2,\alpha=\bar\gamma_{\rm E}\bar\kappa\}$, and
$\{k_1,k_2,\alpha=\lambda\}$ while the choice of $c$ is restricted
by~(\ref{eq:neq0}). Because of their role in ensuring unique
solvability of \eqref{eq:EHsys}, we refer to $\{\lambda,\gamma_{\rm
  E},\gamma_{\rm M},c\}$ as {\it uniqueness parameters}.

Generally, there are many parameter choices for which the conditions
of Proposition~\ref{prop:uniqueB0} and~(\ref{eq:neq0}) hold for a
given $\{k_1,k_2\}$ satisfying \eqref{eq:k1k2}. A valid choice, which
also works well numerically, when $\Arg(k_1)=0$ and
$\pi/4\leq\Arg(k_2)\leq\pi/2$ is
\begin{equation}
\lambda=e^{-{\rm i}\Arg(k_2)},\quad
\gamma_{\rm E}=\kappa^{-1}e^{{\rm i}(\Arg(k_2)-\pi)},\quad
\gamma_{\rm M}=1\,,\quad
c=\lambda^{-1},
\label{eq:parameter2}
\end{equation}
and when $\Arg(k_1)=0$ and $\pi/2\leq\Arg(k_2)\leq 3\pi/4$
\begin{equation}
\lambda=e^{{\rm i}(\pi-\Arg(k_2))},\quad
\gamma_{\rm E}=\kappa^{-1}e^{{\rm i}\Arg(k_2)},\quad
\gamma_{\rm M}=1\,,\quad
c=\lambda^{-1}.
\label{eq:parameter1}
\end{equation}

\section{2D limits}
\label{sec:twoD}

As a first numerical test of our formulations we consider, in
Section~\ref{sec:numex}, the 2D transverse magnetic (TM) transmission
problem where an incident TM wave is scattered from an infinite
cylinder. This problem is independent of the $z$-coordinate and we
introduce the vector $r=(x,y)$, the unit tangent vector
$\tau=(\tau_x,\tau_y)$, and the unit normal vector
$\nu=(\nu_x,\nu_y)$, where $(\tau_x,\tau_y,0)=\hat{\myvec
  z}\times(\nu_x,\nu_y,0)$ and $\hat{\myvec z}$ is the unit vector in
the $z$-direction. The incident wave has polarization $\myvec H^{\rm
  in}(r)=\hat{\myvec z}H^{\rm in}(r)$, which implies $\myvec M_{\rm
  s}=\hat{\myvec z}M$, $\myvec J_{\rm s}=\tau J$, $\varrho_{\rm M}=0$,
and $\sigma_{\rm M}=0$.

The integral representations \eqref{eq:U1}, \eqref{eq:U2}, and
\eqref{eq:E1rep1}--\eqref{eq:H2rep11}, as well as the systems
\eqref{eq:kombi} and \eqref{eq:EHsys}, are transferred to two
dimensions by exchanging the fundamental solution \eqref{eq:fund} for
the 2D fundamental solution
\begin{equation}
\Phi_k(r,r')=\frac{\rm i}{4}H_0^{(1)}(k|r-r'|)\,,\quad r,r'\in\mathbb{R}^2,
\label{eq:Phi}
\end{equation}
where $H_0^{(1)}$ is the zeroth order Hankel function of the first kind.

\subsection{Integral representations in two dimensions}
\label{sec:intrep2D}

Since $\sigma_{\rm E}$ is zero, see Remark \ref{rem:div}, the 2D
representation of the field $\myvec H$ in \eqref{eq:repEH}, to be used
in evaluation of the magnetic field, is
\begin{equation}
H(r)=\left\{
\begin{array}{ll}
\dfrac{1}{2}\tilde S_{k_1}M(r)-\dfrac{1}{2} K_{k_1}J(r)+H^{\rm in}(r)\,,& 
r\in\Omega_1\,,\\
-\dfrac{\kappa}{2}\tilde S_{k_2}M(r)+\dfrac{1}{2}K_{k_2}J(r)\,,&
r\in\Omega_2\,.
\end{array}\right.
\label{eq:H22D}
\end{equation}
By letting $U=H$, $U^{\rm in}=H^{\rm in}$, $\mu=-J$, $\varrho=-{\rm
  i}k_1M$, and $\kappa=k_2^2/k_1^2$ in the scalar representation
\eqref{eq:U} it becomes identical to \eqref{eq:H22D}. According to
Section~\ref{sec:eval} one may add null-fields to \eqref{eq:H22D}.
That gives the representation
\begin{equation}
H(r)=\frac{1}{2}(\tilde S_{k_1}-\kappa\tilde S_{k_2})M(r)
-\frac{1}{2}(K_{k_1}- K_{k_2})J(r)+H^{\rm in}(r)\,,
\quad r\in\Omega_1\cup \Omega_2,
\label{eq:Hnull}
\end{equation}
which is to prefer for evaluations at points $r$ close to $\Gamma$.

\subsection{Integral equations with four, three, and two densities}

In the TM problem the system \eqref{eq:EHsys} becomes
\begin{equation}
\left(I+\tilde{\bf D}\tilde{\bf Q}\right)\tilde{\myvec\mu}
=2\tilde{\bf D}\tilde{\myvec f}\,.
\label{eq:EHsysTM}
\end{equation}
Here $\tilde{\myvec\mu}$ and $\tilde{\myvec f}$ are column vectors
with four entries each
\begin{align*}
\tilde{\myvec\mu}&=\left[\sigma_{\rm E};\varrho_{\rm E};M;J\right],\\
\tilde{\myvec f}&=\left[0;{\rm i}k_1^{-1}\partial_\tau H^{\rm in};
{\rm i}k_1^{-1}\partial_\nu H^{\rm in};-H^{\rm in}\right],
\end{align*}
$\tilde{\bf Q}$ is a $4\times 4$ matrix with non-zero scalar operator
entries
\begin{align*}
\tilde Q_{11}&=-K_{k_1}+c_3K_{k_2}\,,\quad
\tilde Q_{12} =-\tilde S_{k_1}+c_3\kappa\tilde S_{k_2}\,,\quad
\tilde Q_{14} =-C_{k_1}+c_3\kappa C_{k_2}\,,\\
\tilde Q_{21}&=-(\tilde S_{k_1}-c_4\tilde S_{k_2})\nu\cdot\nu'\,,\quad
\tilde Q_{22} =K_{k_1}^{\rm A}-c_4K_{k_2}^{\rm A}\,,\\
\tilde Q_{23}&=C_{k_1}^{\rm A}-c_4\kappa C_{k_2}^{\rm A}\,,\quad
\tilde Q_{24} =-(\tilde S_{k_1}-c_4\kappa \tilde S_{k_2})\nu\cdot\tau'\,,\\
\tilde Q_{31}&=(\tilde S_{k_1}-c_5\kappa^{-1}\tilde S_{k_2})
\tau\cdot\nu'\,,\quad
\tilde Q_{32} =-C_{k_1}^{\rm A}+c_5\kappa^{-1}C_{k_2}^{\rm A}\,,\\
\tilde Q_{33}&=K_{k_1}^{\rm A}-c_5K_{k_2}^{\rm A}\,,\quad
\tilde Q_{34} =(\tilde S_{k_1}-c_5\tilde S_{k_2})\tau\cdot\tau'\,,\\
\tilde Q_{41}&=C_{k_1}-c_6\kappa^{-1}C_{k_2}\,,\quad
\tilde Q_{43} =\tilde S_{k_1}-c_6\kappa \tilde S_{k_2}\,,\quad
\tilde Q_{44} =-K_{k_1}+c_6 K_{k_2}\,,
\end{align*}
$\tilde{\bf D}$ is a diagonal $4\times 4$ matrix of scalars with
non-zero entries
\begin{equation*}
\tilde D_{ii}=(1+c_{i+2})^{-1}\,,\quad i=1,2,3,4\,,
\end{equation*}
and
\begin{equation}
\begin{split}
C_k\sigma(r)&=
2\int_\Gamma (\partial_{\tau'}\Phi_k)(r,r')\sigma(r')\,{\rm d}\Gamma',
\quad r\in\Gamma\,,\\
C_k^{\rm A}\sigma(r)&=
2\int_\Gamma (\partial_{\tau}\Phi_k)(r,r')\sigma(r')\,{\rm d}\Gamma',
\quad r\in\Gamma.
\end{split}
\end{equation}

If we omit $\sigma_{\rm E}$, see Remark \ref{rem:div}, the system
\eqref{eq:EHsysTM} reduces to
\begin{equation}
\left(I+\hat{\bf D}\hat{\bf Q}\right)\hat{\myvec\mu}
=2\hat{\bf D}\hat{\myvec f}\,.
\label{eq:threedens}
\end{equation}
Here $\hat{\bf Q}$ and $\hat{\bf D}$ are $\tilde{\bf Q}$ and
$\tilde{\bf D}$ with the first row and column deleted, $\hat{\myvec
  f}$ is $\tilde{\myvec f}$ with the first entry deleted, and
$\hat{\myvec\mu}$ contains the three densities $\{\varrho_{\rm
  E},M,J\}$.

A third alternative is to only use the densities $M$ and $J$. The
integral representation \eqref{eq:U} and system \eqref{eq:kombi} are
now suitable, where the change of variables in
Section~\ref{sec:intrep2D} makes \eqref{eq:U} equal to \eqref{eq:H22D}
and \eqref{eq:kombi} equal to
\begin{equation}
\begin{bmatrix}I+\beta_2(K_{k_1}^{\rm A}-c_2K_{k_2}^{\rm A})
&\beta_2{\rm i}k_1^{-1}(T_{k_1}-c_2\kappa^{-1}T_{k_2})\\
\beta_1(\tilde S_{k_1}-c_1\kappa\tilde S_{k_2})&I-\beta_1(K_{k_1}-c_1K_{k_2})\end{bmatrix}
\begin{bmatrix}
M\\
J
\end{bmatrix}=
2\begin{bmatrix}
\beta_2{\rm i}k_1^{-1}\partial_\nu H^{\rm in}\\
-\beta_1H^{\rm in}
\end{bmatrix}.
\label{eq:kombi2}
\end{equation}
If the conditions in Theorem~\ref{thm:exunA} hold, then
\eqref{eq:kombi2} has a unique solution $\{M,J\}$. Via \eqref{eq:H22D}
it represents the unique solution to the 2D TM problem.

\section{Test domains and discretization}
\label{sec:geomdisc}

This section reviews domains and discretization schemes that are used
for numerical tests in the next section.

\begin{figure}[t]
\centering
\includegraphics[height=35mm]{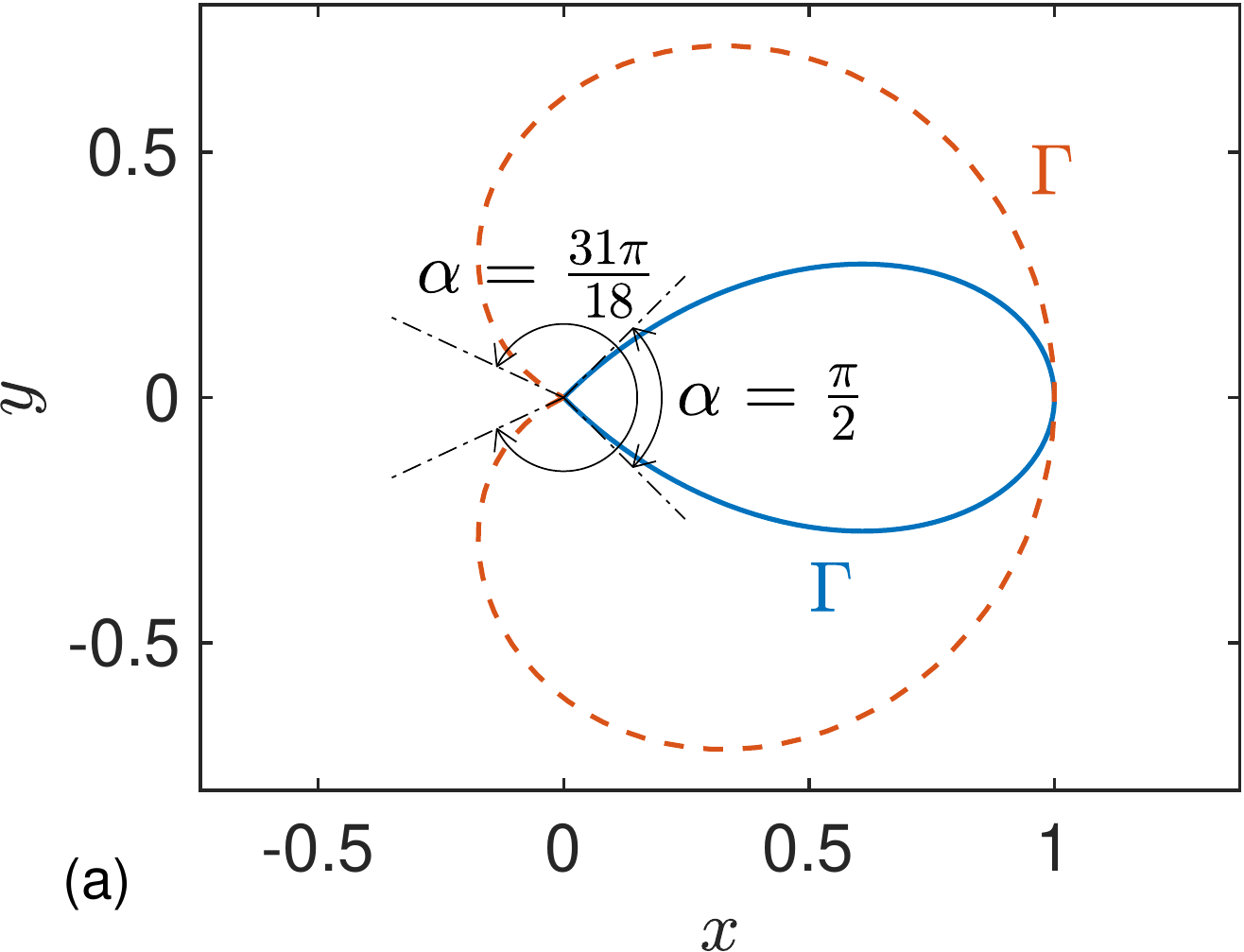}
\includegraphics[height=35mm]{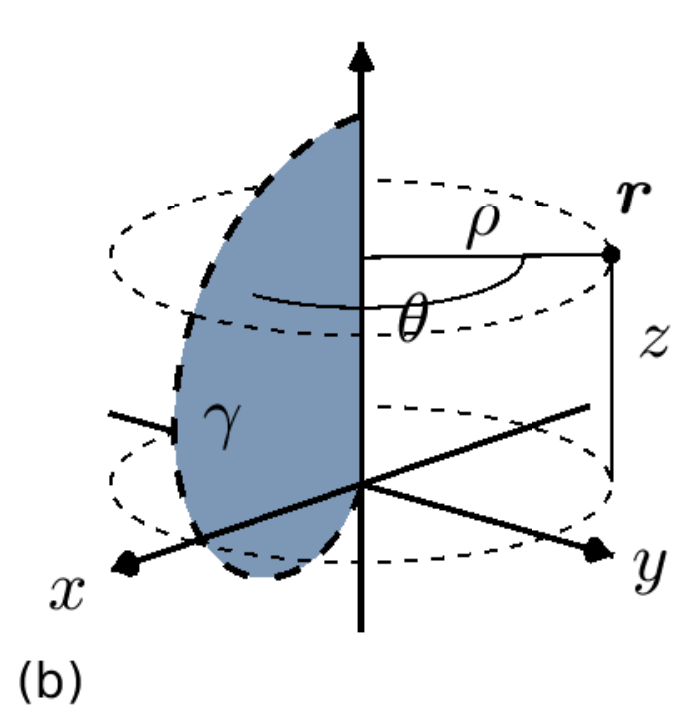}
\includegraphics[height=35mm]{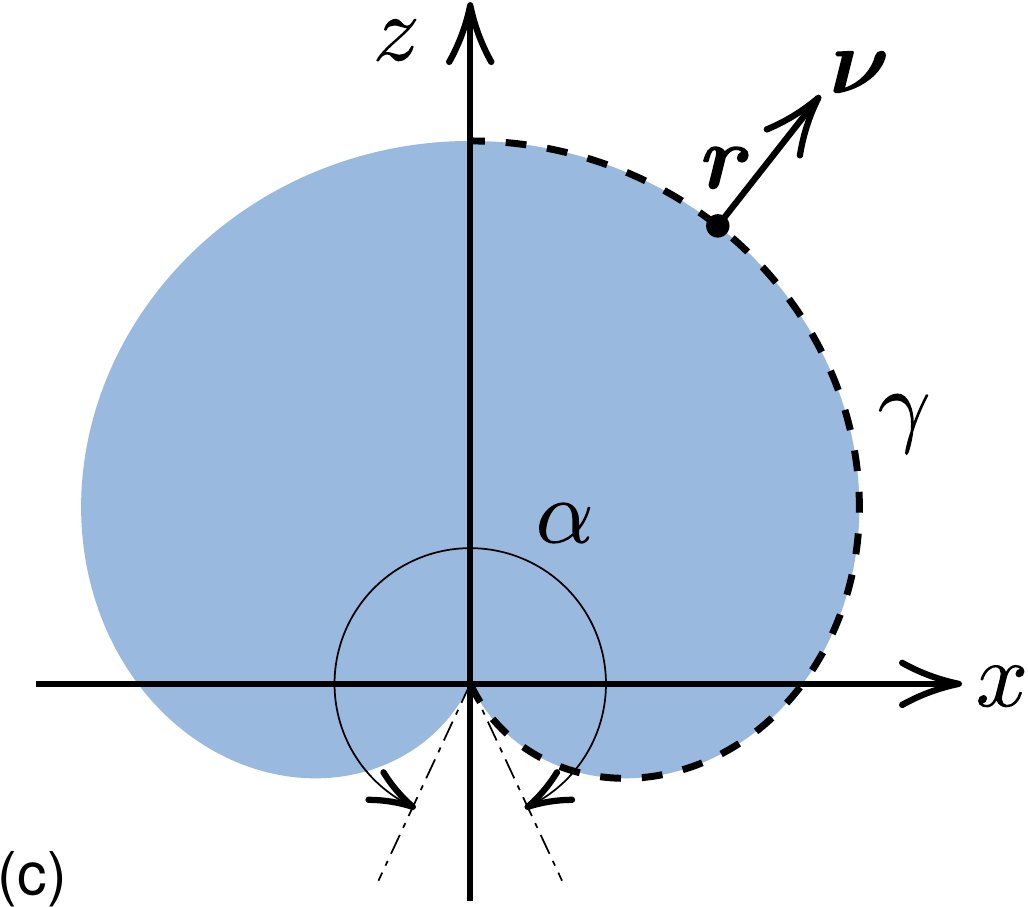}
\caption{\sf Non-smooth test domains: (a) boundaries $\Gamma$ of
  2D domains with corner opening angles $\alpha=\pi/2$ (solid blue)
  and $\alpha=31\pi/18$ (dashed orange); (b) cylindrical coordinates
  $(\rho,\theta,z)$ of a point $\myvec r$ on the surface of an
  axisymmetric object with generating curve $\gamma$; (c) cross
  section of the object generated by $\gamma$ with conical point
  opening angle $\alpha=31\pi/18$.}
\label{fig:amoeba12}
\end{figure}

\subsection{The 2D one-corner object and the 3D ``tomato''}
\label{sec:droptoma}

Numerical tests in two dimensions involve a one-corner object whose
boundary $\Gamma$ is parameterized as
\begin{equation}
r(s)=\sin(\pi s)\left(\cos((s-0.5)\alpha),\sin((s-0.5)\alpha)\right)\,,
\quad s\in[0,1]\,,
\label{eq:gamma2D}
\end{equation}
and where $\alpha$ is a corner opening angle. See
Figure~\ref{fig:amoeba12}(a) for illustrations. 
  
Numerical tests in three dimensions involve an object whose surface
$\Gamma$ is created by revolving the generating curve $\gamma$,
parameterized as
\begin{equation}
\myvec r(s)=
\sin(\pi s)\left(\sin((0.5-s)\alpha),0,\cos((0.5-s)\alpha)\right)\,,
\quad s\in[0,0.5]\,,
\label{eq:gamma3D}
\end{equation}
around the $z$-axis. For $\alpha>\pi$ this object resembles a
``tomato''. See Figure~\ref{fig:amoeba0} and
Figure~\ref{fig:amoeba12}(b,c) for illustrations with
$\alpha=31\pi/18$. 

The reason for testing integral equations in axisymmetric domains,
rather than in general domains, is the availability of efficient
high-order solvers. Use of axisymmetric domains and solvers as a
robust test-bed for new integral equation reformulations of scattering
problems is contemporary common
practice~\cite{EpsGreNei19,LaiOneil19}.

\subsection{RCIP-accelerated Nyström discretization schemes}
\label{sec:disc}

Nyström discretization, relying on composite Gauss--Legendre
quadrature, is used for all our systems of integral equations. Large
discretized linear systems are solved iteratively using GMRES. In the
presence of singular boundary points which call for intense mesh
refinement, the Nyström scheme is accelerated by recursively
compressed inverse preconditioning (RCIP)~\cite{Hels18}. The RCIP acts
as a fully automated, geometry-independent, and fast direct local
solver and boosts the performance of the original Nyström scheme to
the point where problems on non-smooth $\Gamma$ are solved with the
same ease as on smooth $\Gamma$. Accurate evaluations of layer
potentials close to their sources on $\Gamma$ are accomplished using
variants of the techniques first presented in~\cite{Hels09}.

The schemes used in the numerical examples are not entirely new. For
2D problems we use the scheme in~\cite[Section~11.3]{HelsKarl19},
relying on 16-point composite quadrature. For 3D problems we use a
modified unification of the schemes in~\cite{HelsKarl17}
and~\cite{HelsPerf18}, relying on 32-point composite quadrature. A key
feature in the scheme of~\cite{HelsKarl17} is an FFT-accelerated
separation of variables, pioneered by~\cite{YouHaoMar12} and used also
in~\cite{EpsGreNei19,LaiOneil19}.

An important technique in~\cite{HelsKarl17} is the split of the
numerator in $\Phi_k(\myvec r,\myvec r')$ of~(\ref{eq:fund}) into
parts that are even and odd in $\lvert\myvec r-\myvec r'\rvert$.  Let
$G(k,{\myvec r},{\myvec r}')$ be one of the $2\pi$-periodic kernels of
Section~\ref{sec:laypot}. Azimuthal Fourier coefficients
\begin{equation}
G_n=
\frac{1}{\sqrt{2\pi}}\int_{-\pi}^{\pi}e^{-{\rm i}n(\theta-\theta')}
G(k,{\myvec r},{\myvec r}')\,{\rm d}(\theta-\theta')\,,
\quad n=0,\pm 1,\pm 2,\ldots\,,
\label{eq:GF}
\end{equation}
are, for $\myvec r$ and $\myvec r'$ close to each other, computed in
different ways depending on the parity of these parts. When $\Im{\rm
  m}\{k\}$ is small, the split
\begin{equation}
e^{{\rm i}k\lvert\myvec r-\myvec r'\rvert}=
        \cos(k\lvert\myvec r-\myvec r'\rvert)
+{\rm i}\sin(k\lvert\myvec r-\myvec r'\rvert)
\label{eq:split1}
\end{equation}
is efficient for $\Phi_k(\myvec r,\myvec r')$. When $\Im{\rm m}\{k\}$
is large, the terms on the right hand side of~(\ref{eq:split1}) can be
much larger in modulus than the function on the left hand side. Then
numerical cancellation takes place. To fix this problem for large
$\Im{\rm m}\{k\}$, not encountered in~\cite{HelsKarl17}, we introduce
a bump-like function
\begin{equation}
\chi(k,\lvert\myvec r-\myvec r'\rvert)=
e^{-\left(\Im{\rm m}\{k\}{\lvert\myvec r-\myvec r'\rvert/4.6}\right)^8}\,, 
\end{equation}
modify the split~(\ref{eq:split1}) to
\begin{equation}
e^{{\rm i}k\lvert\myvec r-\myvec r'\rvert}=
(1-\chi)e^{{\rm i}k\lvert\myvec r-\myvec r'\rvert}
+\chi\cos(k\lvert\myvec r-\myvec r'\rvert)
+{\rm i}\chi\sin(k\lvert\myvec r-\myvec r'\rvert)\,,
\label{eq:split2}
\end{equation}
and compute $G_n$ of~(\ref{eq:GF}) with techniques (direct transform
or convolution) appropriate for parts of $G(k,{\myvec r},{\myvec r}')$
associated with each of the terms on the right hand side
of~(\ref{eq:split2}).

\section{Numerical examples}
\label{sec:numex}

The systems~(\ref{eq:EHsys}), (\ref{eq:EHsysTM}),
(\ref{eq:threedens}), and (\ref{eq:kombi2}) and the representations
(\ref{eq:repEH}), (\ref{eq:H22D}), and (\ref{eq:Hnull}) are now put to
the test. In all examples we take $k_1$ real and positive,
$\varepsilon_1=1$, and $\varepsilon_2=-1.1838$. This parameter
combination satisfies the plasmonic condition \eqref{eq:plasmonic} and
has been used in previous work on 2D surface plasmon
waves~\cite{Annsop16,HelsKarl18,HelsKarl19}. In situations involving
non-smooth surfaces, it may happen that solutions for
$\varepsilon_2=-1.1838$ do not exist. We then compute limit solutions
as $\varepsilon_2$ approaches $-1.1838$ from above in the complex
plane. Such limit solutions, discussed in the context of Laplace
transmission problems in~\cite[Section~2.2]{HelsPerf18}, have boundary
traces that may best be characterized as lying in fractional-order
Sobolev spaces~\cite{HelsRose20} and are given a downarrow
superscript. For example, the limit of the field $\myvec H$ is denoted
$\myvec H^\downarrow$. The uniqueness parameters
$\{\lambda,\gamma_{\rm E},\gamma_{\rm M},c\}$, needed in
(\ref{eq:EHsys}), (\ref{eq:EHsysTM}), and (\ref{eq:threedens}), are
chosen according to~(\ref{eq:parameter1}). The uniqueness parameters
needed in~(\ref{eq:kombi2}) are chosen as $\{c_1,c_2\}=\{-{\rm
  i},\kappa\}$.

Our codes are implemented in {\sc Matlab}, release 2018b, and executed
on a workstation equipped with an Intel Core i7-3930K CPU and 64 GB of
RAM. When assessing the accuracy of computed field quantities we most
often adopt a procedure where to each numerical solution we also
compute an overresolved reference solution, using roughly 50\% more
points in the discretization of the system under study. The absolute
difference between these two solutions is denoted the {\it estimated
  absolute error}. Throughout the examples, field quantities are
computed at $10^6$ field points on a rectangular Cartesian grid in the
computational domains shown in the figures.

\subsection{Unique solvability on the unit circle}
\label{sec:C2}

We compute condition numbers of the discretized system matrices in
(\ref{eq:EHsysTM}), (\ref{eq:threedens}), and~(\ref{eq:kombi2}). The
boundary $\Gamma$ is the unit circle and $k_1$ is swept through the
interval $[0,10]$. Recall that the systems~(\ref{eq:EHsysTM})
and~(\ref{eq:kombi2}) are guaranteed to be free from wavenumbers for
which the solution is not unique (false eigenwavenumbers) while the
system (\ref{eq:threedens}) is not.

Condition number analysis of 2D limits of 3D systems on the unit
circle is a revealing test for detecting if a given system of integral
equations has false eigenwavenumbers when the plasmonic condition
holds. For example, in~\cite[Figure~9]{HelsKarl19} it is shown that
the original Müller system and the ``$\myvec E$-system''
of~\cite{VicGreFer18} exhibit several false eigenwavenumbers in such a
test.

\begin{figure}[t]
\centering
   \includegraphics[height=47mm]{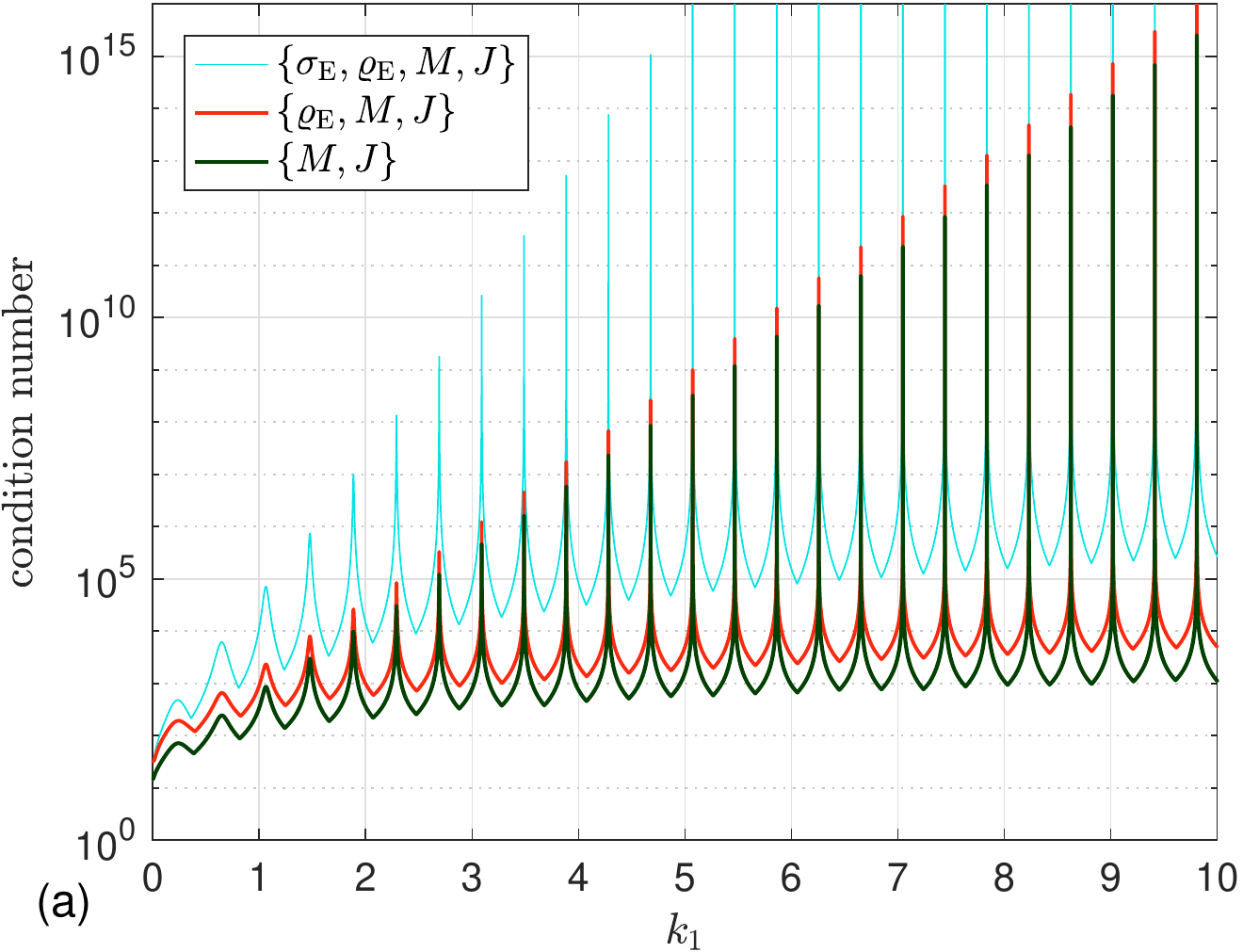}
   \includegraphics[height=47mm]{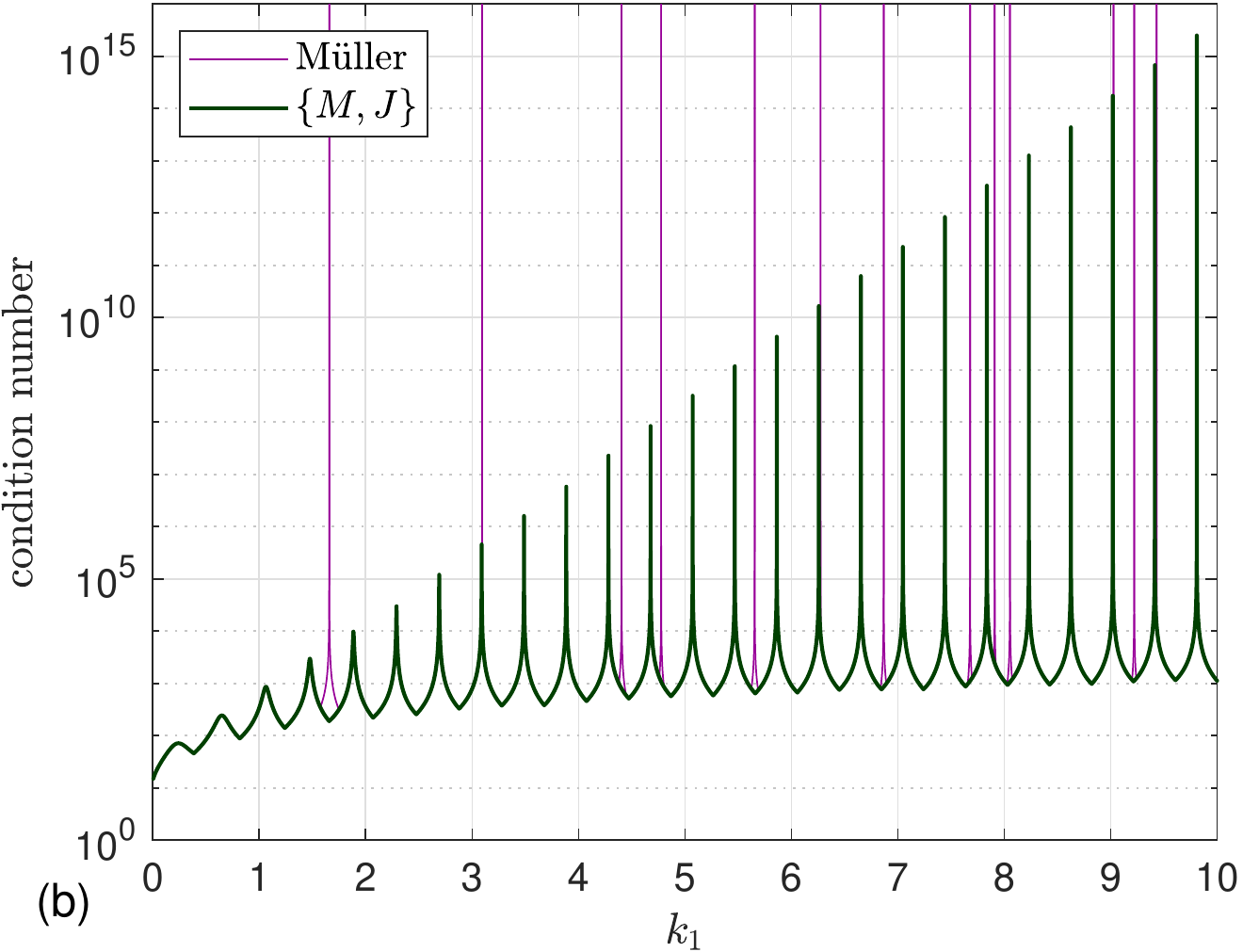}
\caption{\sf Condition numbers of system matrices on the unit
  circle, $\varepsilon_1=1$, $\varepsilon_2=-1.1838$, and
  $k_1\in[0,10]$: (a) the systems~(\ref{eq:EHsysTM}),
  (\ref{eq:threedens}), and~(\ref{eq:kombi2}); (b) the Müller system
  and a repeat of the bottom curve in (a).}
\label{fig:cond2}
\end{figure}

Figure~\ref{fig:cond2}(a) shows results obtained with
(\ref{eq:EHsysTM}), (\ref{eq:threedens}), and~(\ref{eq:kombi2}) using
768 discretizations points on $\Gamma$ and approximately $20,\!700$
values of $k_1\in[0,10]$. The regularly recurring high peaks
correspond to true eigenwavenumbers just below the positive $k_1$-axis
(weakly damped dynamic surface plasmons). One can see that neither the
four-density system (\ref{eq:EHsysTM}) nor the two-density
system~(\ref{eq:kombi2}) exhibits any false eigenwavenumbers, as
expected, and that~(\ref{eq:kombi2}) is the best conditioned system.
Furthermore, which is more remarkable, the three-density system
(\ref{eq:threedens}) also appears to be free from false
eigenwavenumbers. For comparison, Figure~\ref{fig:cond2}(b) shows
condition numbers of the original Müller system, corresponding to
$\{c_1,c_2\}=\{1,\kappa\}$ in \eqref{eq:kombi2}. Here one can see 13
false eigenwavenumbers. Some distance away from these wavenumbers the
results from the Müller system and~(\ref{eq:kombi2}) with
$\{c_1,c_2\}=\{-{\rm i},\kappa\}$ overlap.

\subsection{Field accuracy for the 2D one-corner object}
\label{sec:F2}

An incident plane wave with $\myvec H^{\rm in}(r)=\hat{\myvec
  z}e^{{\rm i}k_1d\cdot r}$, $k_1=18$, and direction of propagation
$d=\left(\cos(\pi/4),\sin(\pi/4)\right)$ is scattered against the 2D
one-corner object of Section~\ref{sec:droptoma}. The corner opening
angle is $\alpha=\pi/2$. A number of $800$ discretization points is
placed on $\Gamma$ and the performance of the three systems
(\ref{eq:EHsysTM}), (\ref{eq:threedens}), (\ref{eq:kombi2}) are
compared.

\begin{figure}[t]
\centering
   \includegraphics[height=46mm]{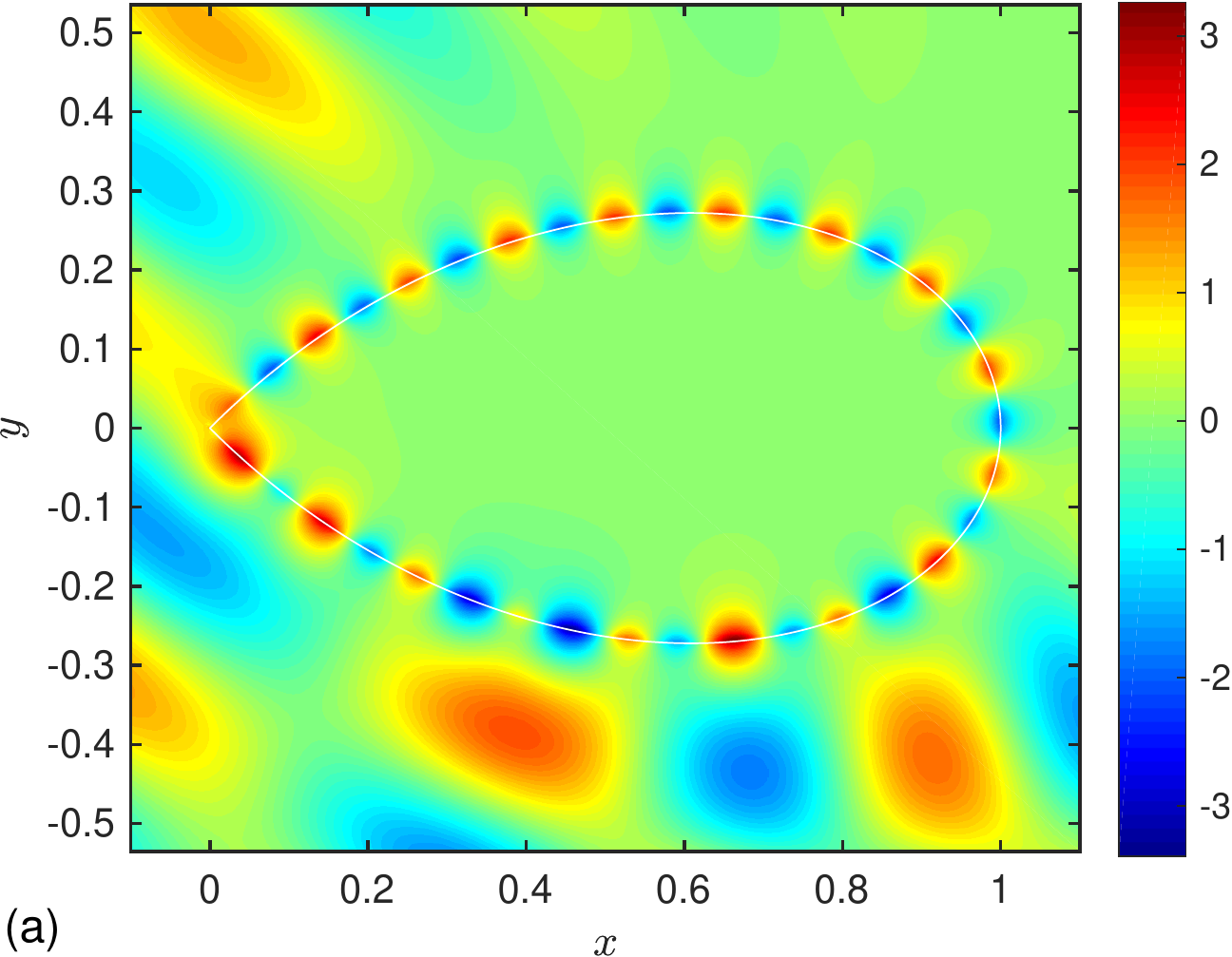}
   \includegraphics[height=46mm]{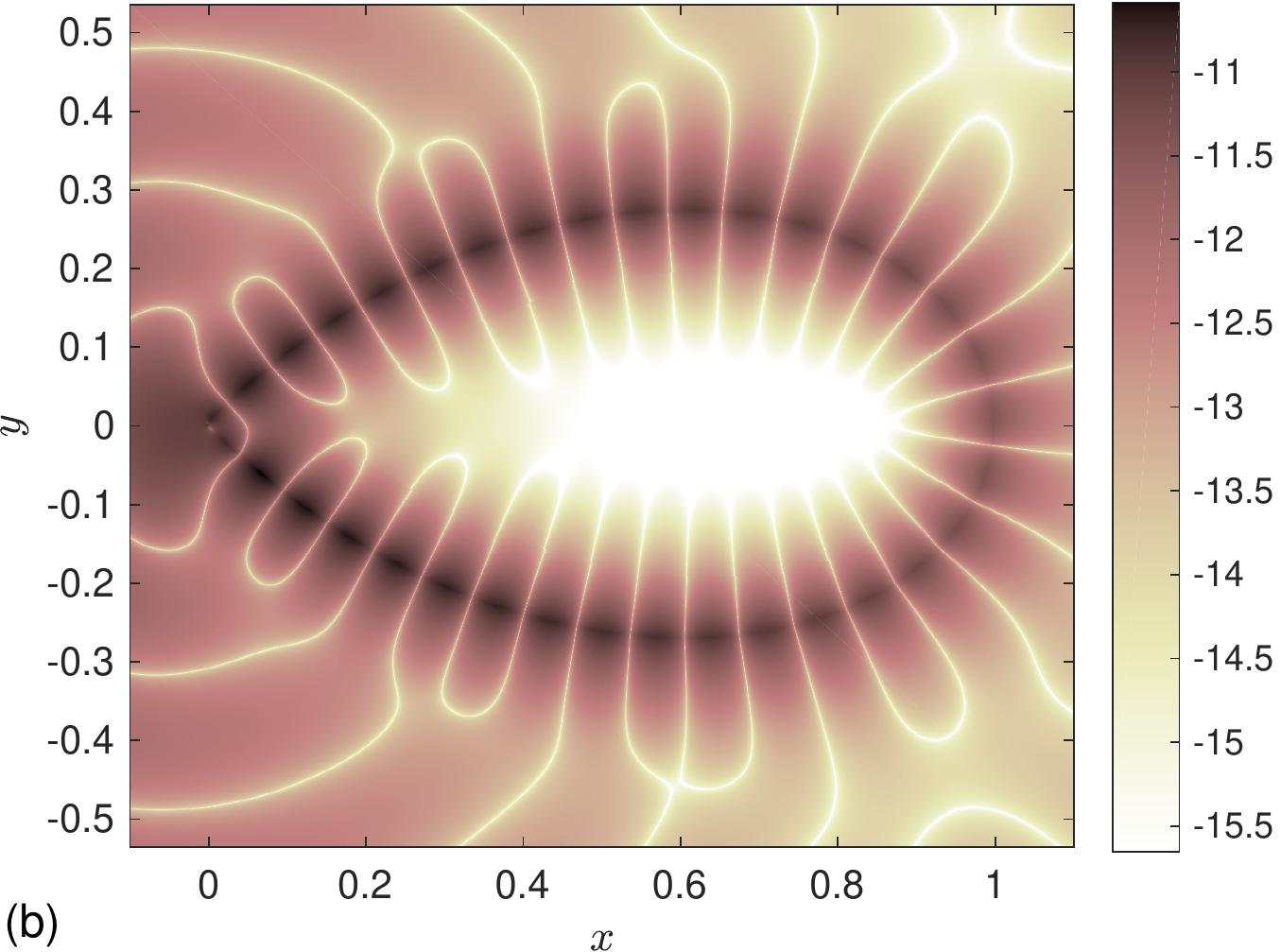}
   \includegraphics[height=46mm]{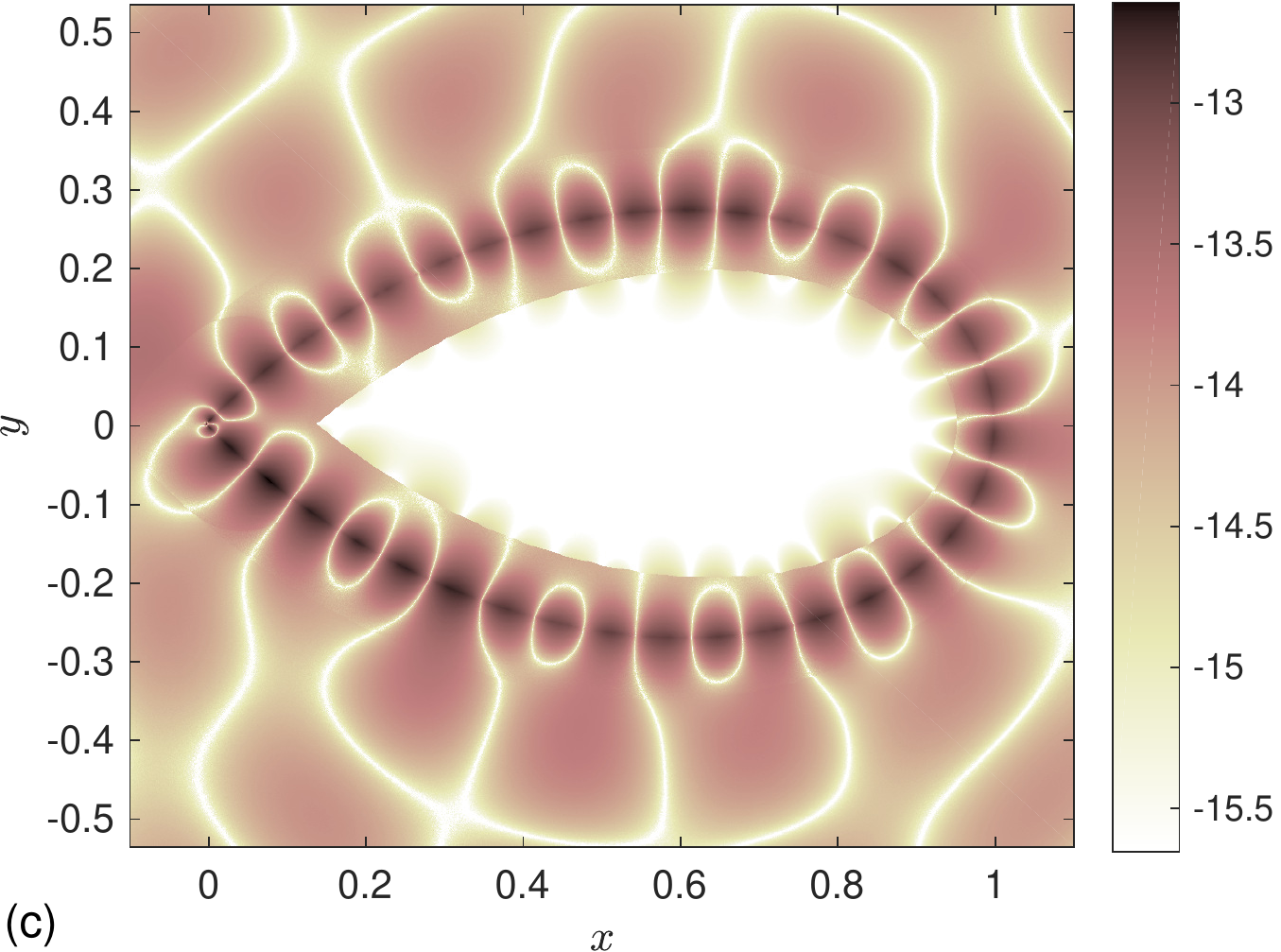}
   \includegraphics[height=46mm]{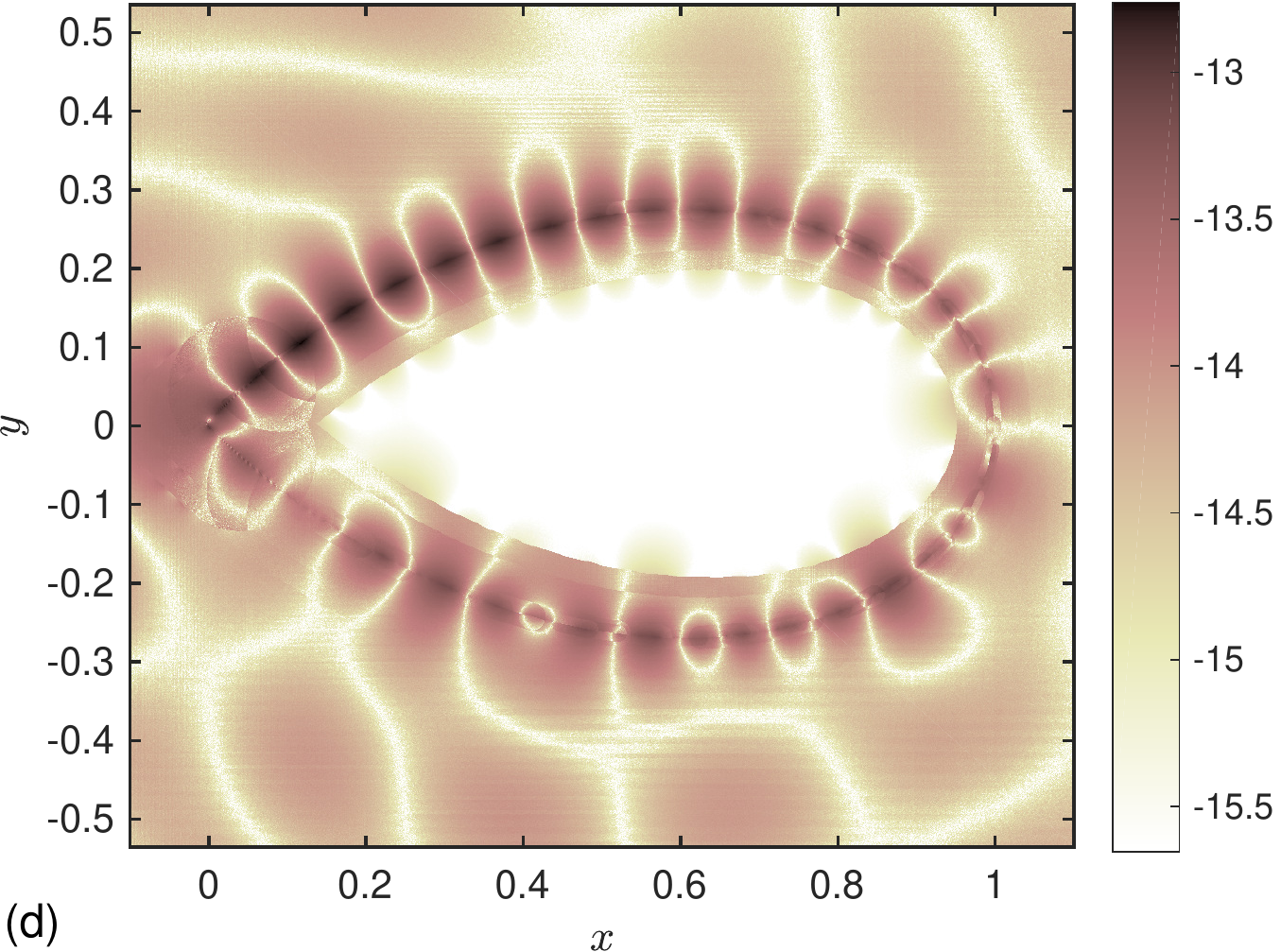}
\caption{\sf The field $H^\downarrow(r,0)$ on the 2D one-corner object with
   $\varepsilon_1=1$, $\varepsilon_2=-1.1838$, and $k_1=18$: (a) the
   field $H^\downarrow(r,0)$ itself; (b,c,d) $\log_{10}$ of estimated
   absolute field error using the systems (\ref{eq:EHsysTM}),
   (\ref{eq:threedens}), and~(\ref{eq:kombi2}), respectively.}
\label{fig:field2}
\end{figure}

Figure~\ref{fig:field2}(a) shows the total magnetic field $\Re{\rm
  e}\{H^\downarrow(r)\}$, see~(\ref{eq:timedep}), and
Figures~\ref{fig:field2}(b,c,d) show $\log_{10}$ of the
estimated absolute error obtained with (\ref{eq:EHsysTM}),
(\ref{eq:threedens}), and (\ref{eq:kombi2}), respectively. The number
of GMRES iterations required to solve the discretized linear systems
is 266 for (\ref{eq:EHsysTM}), 154 for (\ref{eq:threedens}), and 143
for (\ref{eq:kombi2}). The absolute errors for the systems
(\ref{eq:EHsysTM}) and (\ref{eq:threedens}) are estimated using the
solution from (\ref{eq:kombi2}) as reference.

It is interesting to observe, in Figure~\ref{fig:field2}, that the
field accuracy is high for all three systems. The number of digits
lost is in agreement with what could be expected for computations on
the unit circle, considering the condition numbers shown in
Figure~\ref{fig:cond2} and assuming that $k_1$ is not close to a true
eigenwavenumber. Note also that~(\ref{eq:kombi2}) is a system of
Fredholm second-kind integral equations with operator differences that
are compact on smooth $\Gamma$ -- a property often sought for in
integral equation modeling of PDEs. The system (\ref{eq:threedens}),
on the other hand, contains a Cauchy-type singular difference of
operators. Still, the performance of the two systems is very similar.

\subsection{Unique solvability on the unit sphere}
\label{sec:C3}

We repeat the experiment of Section~\ref{sec:C2}, but now on the unit
sphere using the system~(\ref{eq:EHsys}). Inspired by the good
performance of the system~(\ref{eq:threedens}), reported above and
where $\sigma_{\rm E}$ is omitted, we omit both $\sigma_{\rm E}$ and
$\sigma_{\rm M}$ from~(\ref{eq:EHsys}) to get a six-scalar-density
system. Again, there is noo proof that this system has a unique
solution, but every solution to the time harmonic Maxwell's equations
corresponds to a solution to this system.

\begin{figure}[t]
\centering 
  \includegraphics[height=47mm]{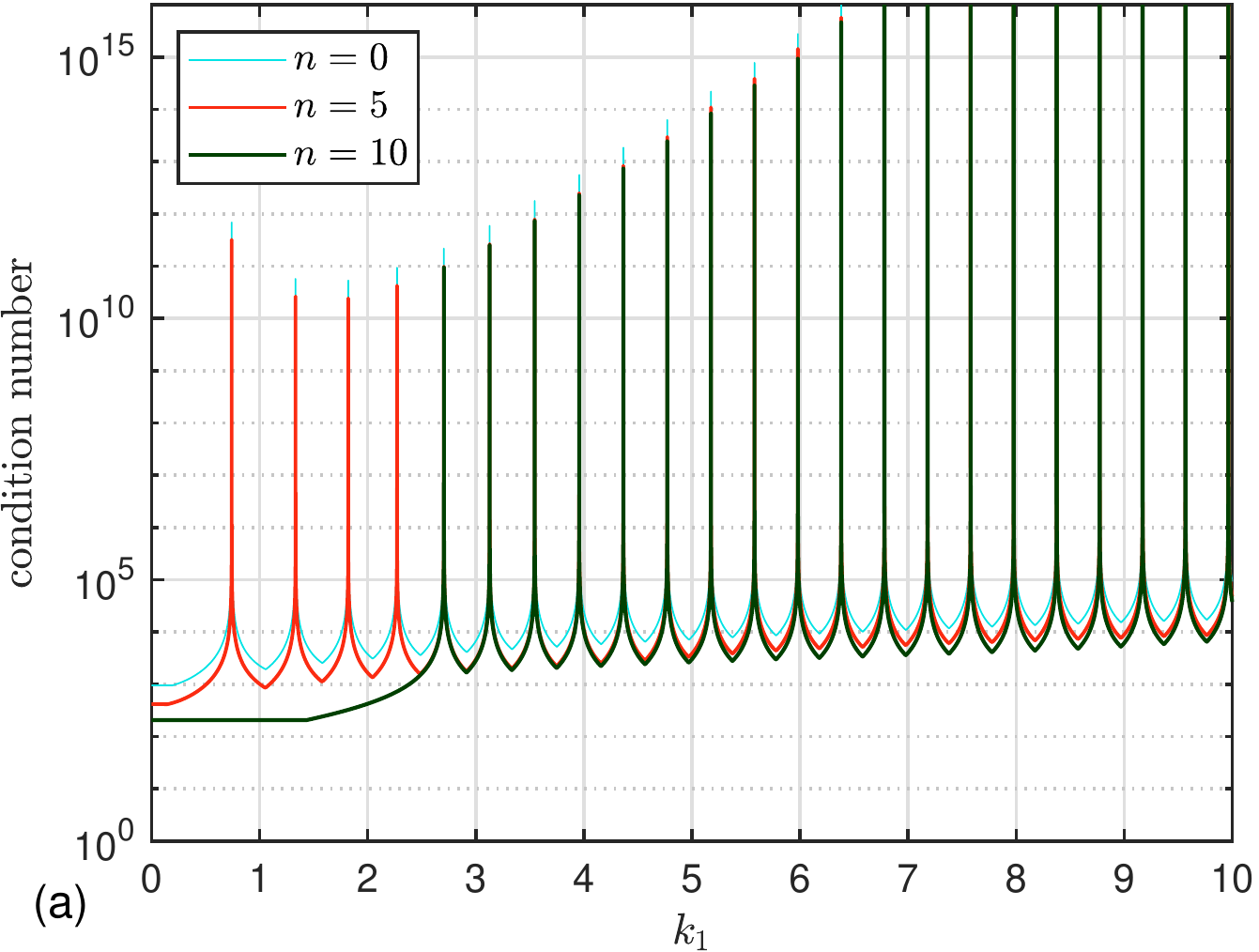}
  \includegraphics[height=47mm]{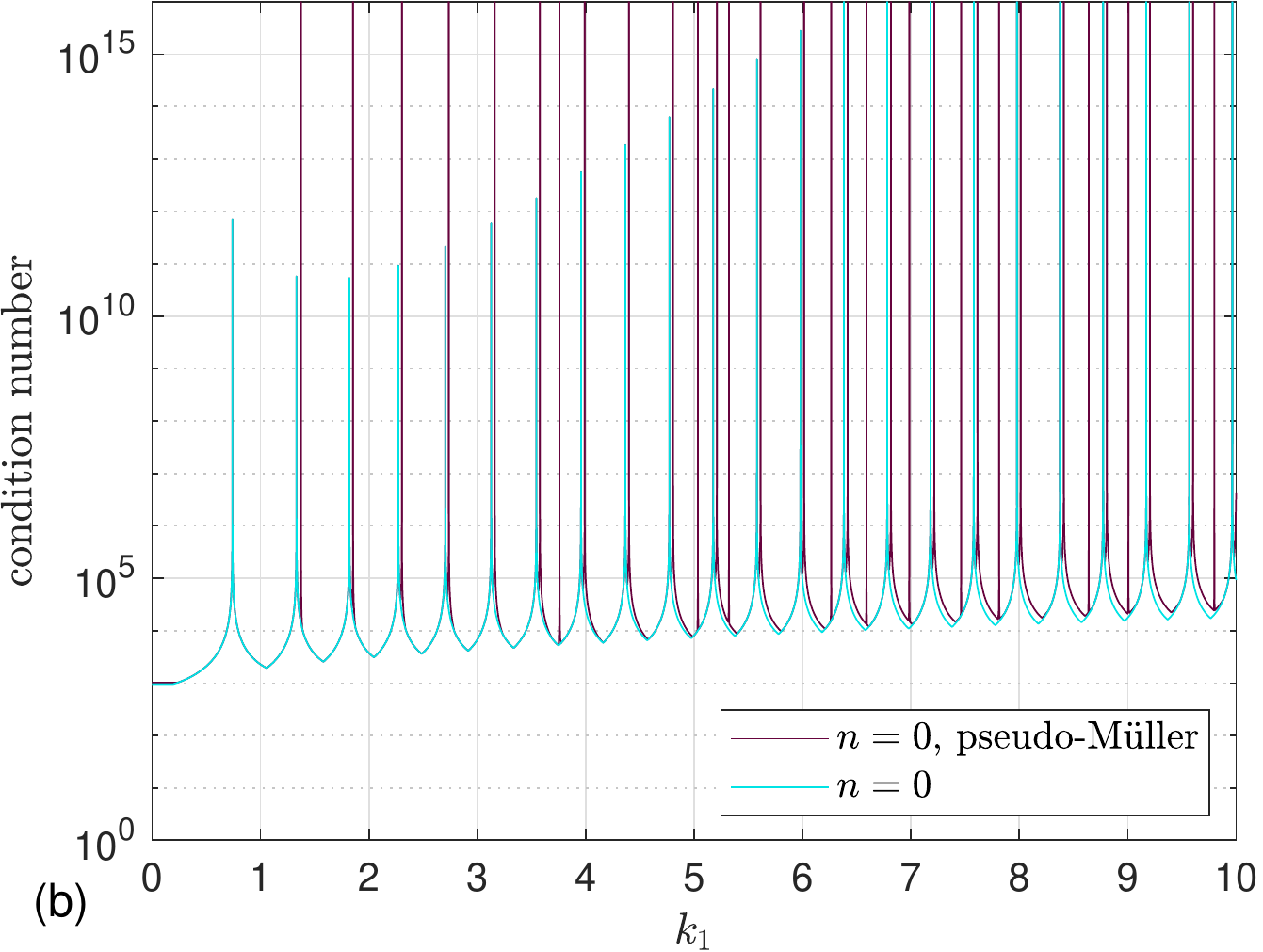}
\caption{\sf Condition numbers of system matrices on the unit 
  sphere, $\varepsilon_1=1$, $\varepsilon_2=-1.1838$, and
  $k_1\in[0,10]$: (a) azimuthal modes $n=0,5,10$ of the
  system~(\ref{eq:EHsys}) with $\sigma_{\rm E}$ and $\sigma_{\rm M}$
  omitted; (b) azimuthal mode $n=0$ of the pseudo-Müller system and a
  repeat of the top curve in (a).}
\label{fig:cond3}
\end{figure}

The Fourier--Nyström scheme of~\cite{HelsKarl17}, see
Section~\ref{sec:disc}, decomposes the reduced system~(\ref{eq:EHsys})
into a sequence of smaller, modal, systems on the generating curve
$\gamma$. Figure~\ref{fig:cond3}(a) shows result for the azimuthal
modes $n=0,5,10$, with 768 discretization points on $\gamma$, and with
approximately $3,\!500$ values of $k_1\in[0,10]$. No false
eigenwavenumbers can be seen. For comparison,
Figure~\ref{fig:cond3}(b) shows results for a six-scalar-density
variant of the Müller system. The original four-scalar-density Müller
system \cite[p. 319]{Muller69} uses the surface current densities
$\myvec M_{\rm s}$ and $\myvec J_{\rm s}$ and contains compact
differences of hypersingular operators. These operator differences are
hard to implement in three dimensions, even though it definitely is
possible on axisymmetric surfaces~\cite{LaiOneil19}. Our variant of
the Müller system is derived from the original Müller system via
integration by parts and relating the surface divergence of $\myvec
M_{\rm s}$ and $\myvec J_{\rm s}$ to $\varrho_{\rm M}$ and
$\varrho_{\rm E}$, see \cite[Eqs.~(36) and~(35)]{HelsKarl17}. This
corresponds to omitting both $\sigma_{\rm E}$ and $\sigma_{\rm M}$
from~(\ref{eq:EHsys}) and setting $c_4=c_6=1$, and $c_5=c_7=\kappa$.
Figure~\ref{fig:cond3}(b) shows that this pseudo-Müller system
exhibits at least $32$ false eigenwavenumbers for $k_1\in[0,10]$.

\subsection{Field accuracy for the 3D ``tomato''}
\label{sec:F3}

An incident linearly polarized plane wave with $\myvec E^{\rm
  in}(\myvec r)=\hat{\myvec x}e^{{\rm i}k_1z}$ and $k_1=5$ is
scattered against the 3D ``tomato'' of Section~\ref{sec:droptoma}. The
conical point opening angle is $\alpha=31\pi/18$. The same
six-scalar-density version of the system~(\ref{eq:EHsys}) is used as
in Section~\ref{sec:C3}. Only two azimuthal modes, $n=-1$ and $n=1$,
are present in this problem and the Fourier coefficients of the
surface densities of these modes are either identical or have opposite
signs. Therefore only one modal system needs to be solved numerically.

\begin{figure}[t]
\centering 
  \includegraphics[height=51mm]{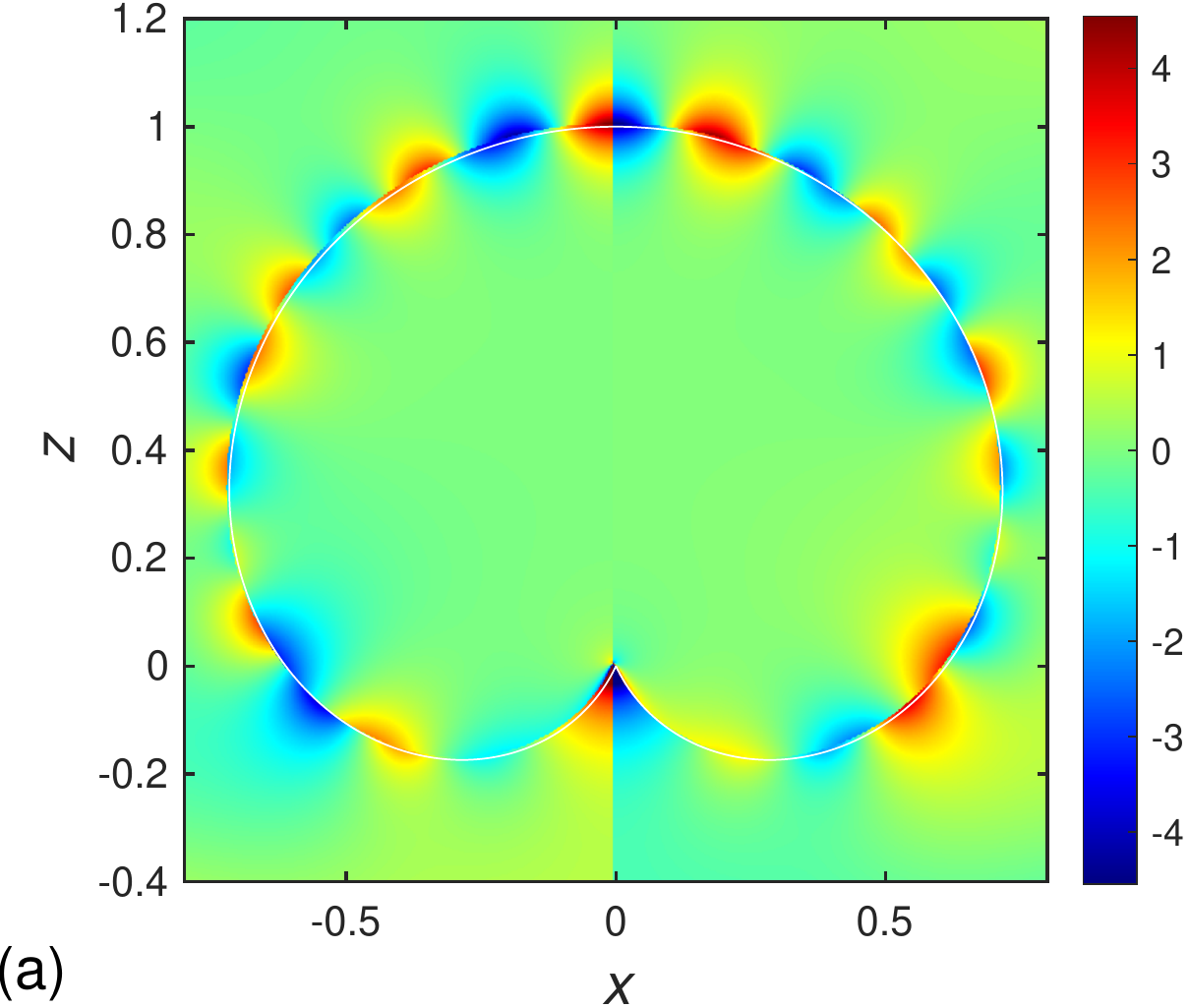}
  \includegraphics[height=51mm]{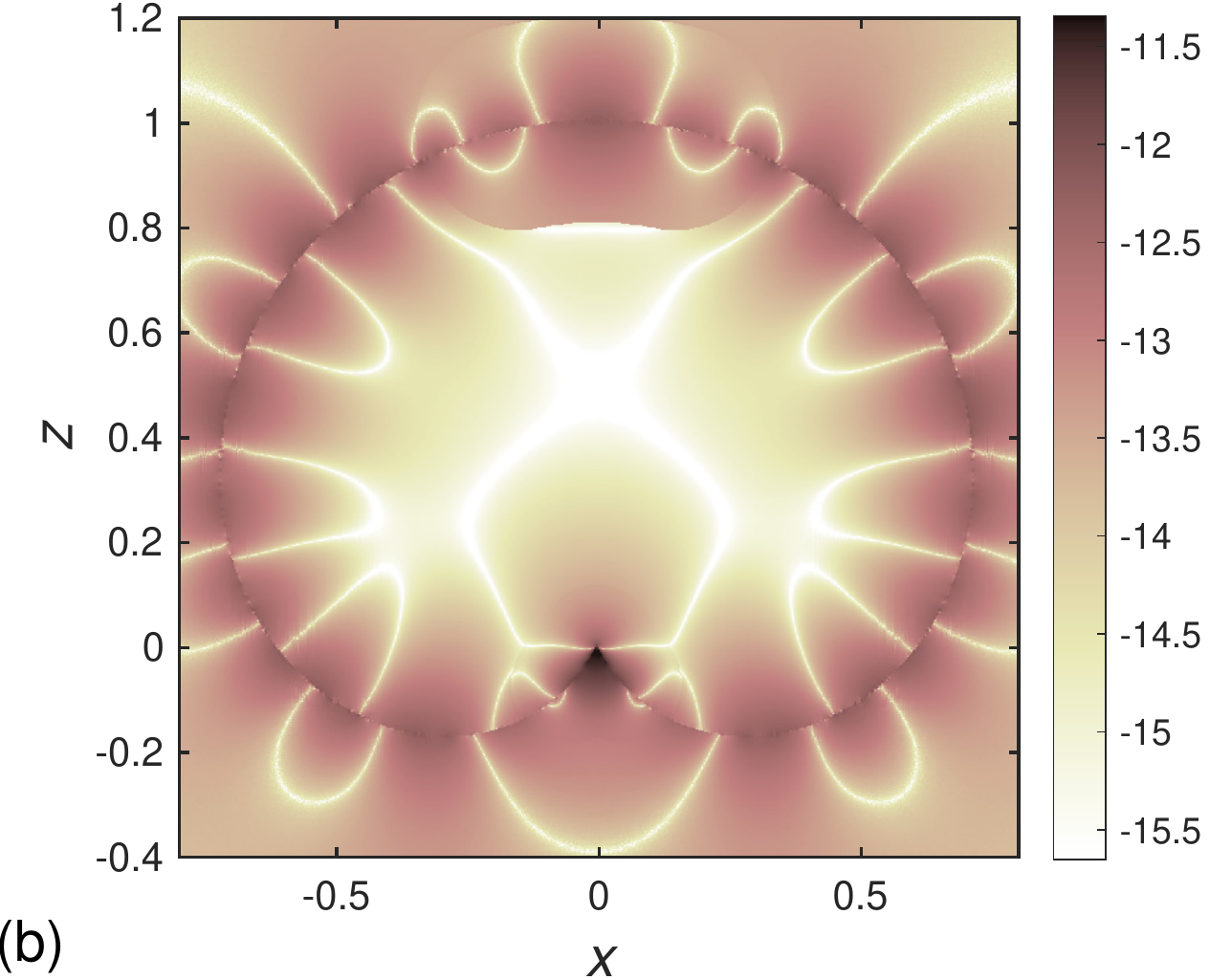}
  \includegraphics[height=51mm]{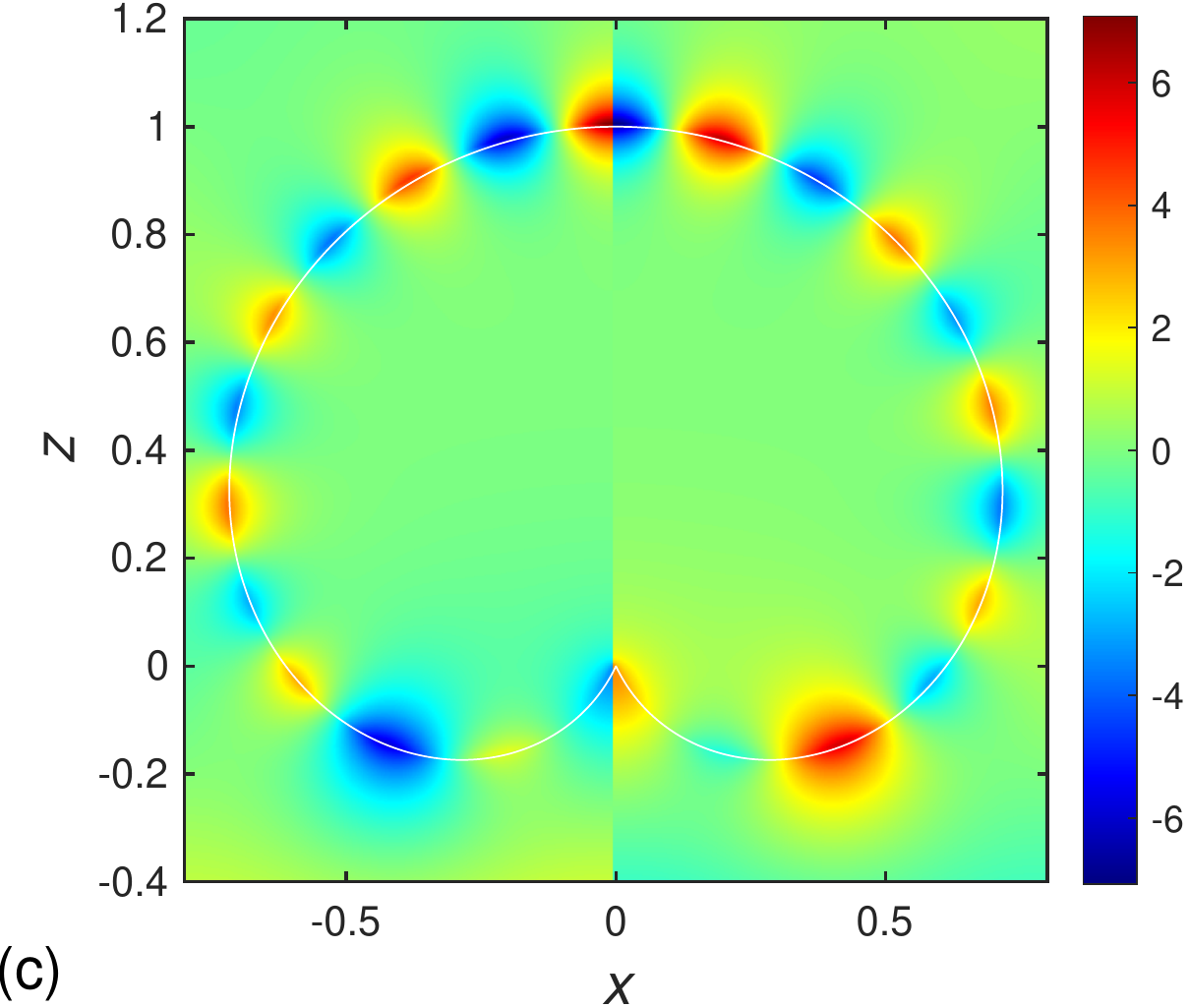}
  \includegraphics[height=51mm]{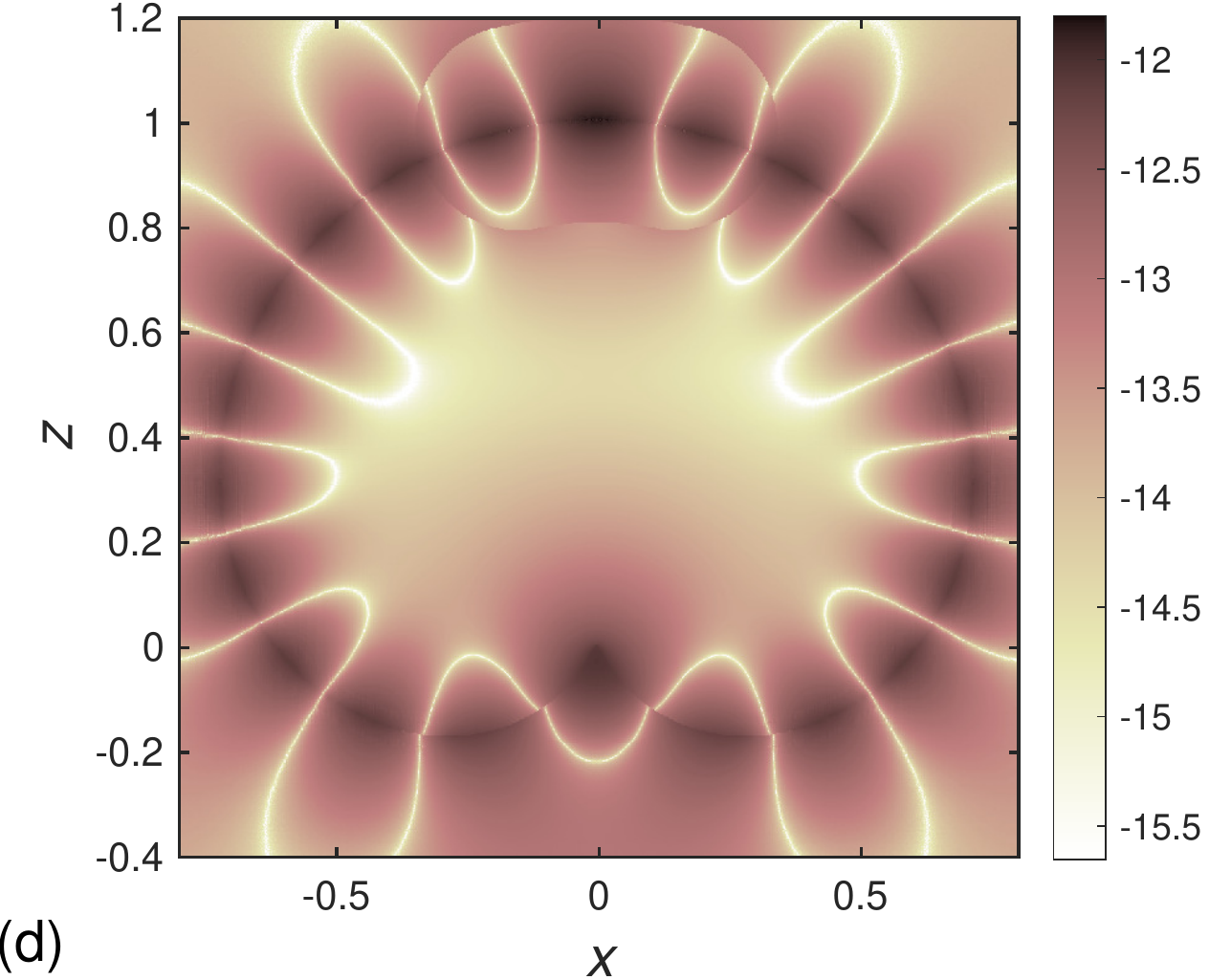}
\caption{\sf Field images on a cross section of the 3D ``tomato'' 
  subjected to an incident plane wave $\myvec E^{\rm in}(\myvec
  r)=\hat{\myvec x} e^{{\rm i}k_1z}$ and with $\varepsilon_1=1$,
  $\varepsilon_2=-1.1838$, and $k_1=5$: (a) the field $E_\rho(\myvec
  r,0)$ with colorbar range set to $[-4.55,4.55]$ ; (b) $\log_{10}$ of
  estimated absolute field error in $E_\rho(\myvec r,0)$; (c) the
  field $H_\theta(\myvec r,0)$; (d) $\log_{10}$ of estimated absolute
  field error in $H_\theta(\myvec r,0)$.}
\label{fig:field3}
\end{figure}

Figure~\ref{fig:field3} shows the electric field in the
$\rho$-direction, $E_\rho(\myvec r,0)$, and the magnetic field in the
$\theta$-direction, $H_\theta(\myvec r,0)$, on the cross section in
Figure~\ref{fig:amoeba12}(c). The results are obtained with $576$
discretization points on the generating curve $\gamma$ and with 242
GMRES iterations. Since the field $E_\rho(\myvec r,0)$ is singular at
the origin, the colorbar range in Figure~\ref{fig:field3}(a) is
restricted to the most extreme values of $E_\rho(\myvec r,0)$ away
from the origin. The precision shown in Figure~\ref{fig:field3}(b,d)
is consistent with the condition numbers of Figure~\ref{fig:cond3}(a)
in the sense discussed in Section~\ref{sec:F2}. We conclude by noting
that Figure~\ref{fig:field3} clearly shows an accurately computed
surface plasmon wave on a non-smooth 3D object in a setup with
negative permittivity ratio. To simulate such surface waves is the
ultimate goal of this work.

\section{Conclusions}

A new system of Fredholm second-kind integral equations is presented
for an electromagnetic transmission problem involving a single
scattering object. Our work can be seen as an extension of the work by
Kleinman and Martin~\cite{KleiMart88} on direct methods for scalar
transmission problems. Thanks to the introduction of certain
uniqueness parameters, our new system gives unique solutions for a
wider range of wavenumber combinations than do other systems of
integral equations for Maxwell's equations, for example the original
Müller system. In particular, unique solutions are guaranteed for
smooth scatterers under the plasmonic condition \eqref{eq:plasmonic},
which is of great interest in physical and engineering applications.

The favorable properties of our new system extend beyond what can be
proven rigorously. In a numerical example, a reduced version of the
system in combination with a high-order Fourier--Nyström
discretization scheme is shown to produce accurate field images of a
surface plasmon wave on a non-smooth axisymmetric scatterer.

\section*{Acknowledgement}

\noindent
We thank Andreas Rosén (formerly Andreas Axelsson) for many useful
conversations. This work was supported by the Swedish Research Council
under contract 621-2014-5159.

\section*{Appendix}

\renewcommand{\theequation}{A.\arabic{equation}}
\setcounter{equation}{0}

\subsection*{A. Boundary limits of $\myvec E$ and  $\myvec H$}

The relations in Section~\ref{sec:limits} give the following limits at
$\Gamma$ for the integral representations of $\myvec E$ and $\myvec H$
in \eqref{eq:E1rep1}--\eqref{eq:H2rep11}:
\begin{align}
[\nabla\cdot\myvec E_1]^\pm&=
\mp\frac{{\rm i}k_1}{2}\sigma_{\rm E}
-\frac{{\rm i}k_1}{2}\tilde{\cal S}_{k_1}\varrho_{\rm E}
+\frac{1}{2}\nabla\cdot\tilde{\cal S}_{k_1}
(\myvec\nu'\sigma_{\rm E}+\myvec J_{\rm s})\,,\\
\myvec\nu\cdot\myvec E_1^\pm&=\pm\frac{1}{2}\varrho_{\rm
  E}-\frac{1}{2}\myvec\nu\cdot\boldsymbol{\cal
  N}_{k_1}\varrho_{\rm E}-\frac{1}{2}\myvec \nu\cdot {\cal
  R}_{k_1}(\myvec \nu'\sigma_{\rm M}+\myvec M_{\rm s})
\nonumber\\
&\qquad\qquad
+\frac{1}{2}\myvec\nu\cdot \tilde{\cal S}_{k_1}
(\myvec\nu'\sigma_{\rm E}+\myvec J_{\rm s})
+\myvec\nu\cdot\myvec E^{\rm in},
\end{align}
\begin{align}
\myvec\nu\times\myvec E_1^\pm&=\mp\frac{1}{2}\myvec M_{\rm
  s}-\frac{1}{2}\myvec \nu\times\boldsymbol{\cal
  N}_{k_1}\varrho_{\rm E}-\frac{1}{2}\myvec \nu\times{\cal
  R}_{k_1}(\myvec\nu'\sigma_{\rm M}+\myvec M_{\rm s})
\nonumber\\
&\qquad\qquad
+\frac{1}{2}\myvec\nu\times\tilde{\cal S}_{k_1}(\myvec
\nu'\sigma_{\rm E}+\myvec J_{\rm s})+\myvec\nu\times\myvec
E^{\rm in},\\
\myvec\nu\times\myvec H_1^\pm&=
\pm\frac{1}{2}\myvec J_{\rm s}
+\frac{1}{2}\myvec\nu\times\tilde{\cal S}_{k_1}
(\myvec\nu'\sigma_{\rm M}+\myvec M_{\rm s})
+\frac{1}{2}\myvec\nu\times{\cal R}_{k_1}
(\myvec\nu'\sigma_{\rm E}+\myvec J_{\rm s})\nonumber\\
&\qquad\qquad
-\frac{1}{2}\myvec\nu\times\boldsymbol{\cal N}_{k_1}
\varrho_{\rm M}+\myvec\nu\times \myvec H^{\rm in},\\
\myvec\nu\cdot\myvec H_1^\pm&=\pm\frac{1}{2}\varrho_{\rm
  M}+\frac{1}{2}\myvec\nu\cdot\tilde{\cal
  S}_{k_1}(\myvec\nu'\sigma_{\rm M}+\myvec M_{\rm s})
+\frac{1}{2}\myvec\nu\cdot{\cal R}_{k_1}
(\myvec \nu'\sigma_{\rm E}+\myvec J_{\rm s})\nonumber\\
&\qquad\qquad
-\frac{1}{2}\myvec\nu\cdot\boldsymbol{\cal N}_{k_1}
\varrho_{\rm M}+\myvec\nu\cdot \myvec H^{\rm in},\\
[\nabla\cdot\myvec H_1]^\pm&=
\mp\frac{{\rm i}k_1}{2}\sigma_{\rm M}
+\frac{1}{2}\nabla\cdot\tilde{\cal S}_{k_1}
(\myvec\nu'\sigma_{\rm M}+\myvec M_{\rm s})
-\frac{{\rm i}k_1}{2}\tilde{\cal S}_{k_1}\varrho_M\,,\\
[\nabla\cdot\myvec E_2]^\pm&=
\pm\frac{{\rm i}k_1}{2\kappa}\sigma_{\rm E}
+\frac{{\rm i}k_1}{2}\tilde{\cal S}_{k_2}\varrho_{\rm E}
-\frac{1}{2}\nabla\cdot\tilde{\cal S}_{k_2}
(\kappa^{-1}\myvec\nu'\sigma_{\rm E}+\myvec J_{\rm s})\,,\\
\myvec\nu\cdot\myvec E_2^\pm&=
\mp\frac{1}{2\kappa}\varrho_{\rm E}
+\frac{1}{2\kappa}\myvec\nu\cdot\boldsymbol{\cal N}_{k_2}
\varrho_{\rm E}
+\frac{1}{2\kappa}\myvec\nu\cdot{\cal R}_{k_2}
(\myvec\nu'\sigma_{\rm M}+\kappa\myvec M_{\rm s})
\nonumber\\
&\qquad\qquad 
-\frac{1}{2}\myvec\nu\cdot\tilde{\cal S}_{k_2}
(\kappa^{-1}\myvec\nu'\sigma_{\rm E}+\myvec J_{\rm s})\,,\\
\myvec\nu\times\myvec E_2^\pm&=
\pm\frac{1}{2}\myvec M_{\rm s}
+\frac{1}{2\kappa}\myvec\nu\times\boldsymbol{\cal N}_{k_2}
\varrho_{\rm E}+\frac{1}{2\kappa}\myvec\nu\times{\cal R}_{k_2}
(\myvec\nu'\sigma_{\rm M}+\kappa\myvec M_{\rm s})
\nonumber\\
&\qquad\qquad
-\frac{1}{2}\myvec\nu\times\tilde{\cal S}_{k_2}
(\kappa^{-1}\myvec\nu'\sigma_{\rm E}+\myvec J_{\rm s})\,,\\
\myvec\nu\times\myvec H_2^\pm&=
\mp\frac{1}{2}\myvec J_{\rm s}
-\frac{1}{2}\myvec\nu\times\tilde{\cal S}_{k_2}
(\myvec\nu'\sigma_{\rm M}+\kappa\myvec M_{\rm s})
\nonumber\\
&\qquad
-\frac{1}{2}\myvec\nu\times{\cal R}_{k_2}
 (\kappa^{-1}\myvec\nu'\sigma_{\rm E}
+\myvec J_{\rm s})
+\frac{1}{2}\myvec\nu\times\boldsymbol{\cal N}_{k_2}
\varrho_{\rm M}\,,\\
\myvec\nu\cdot\myvec H_2^\pm&=
\mp\frac{1}{2}\varrho_{\rm M}
-\frac{1}{2}\myvec\nu\cdot\tilde{\cal S}_{k_2}(\myvec\nu'\sigma_{\rm M}
+\kappa\myvec M_{\rm s})
\nonumber\\
&\qquad
-\frac{1}{2}\myvec\nu\cdot{\cal R}_{k_2}
(\kappa^{-1}\myvec\nu'\sigma_{\rm E}+\myvec J_{\rm s})
+\frac{1}{2}\myvec\nu\cdot\boldsymbol{\cal N}_{k_2}
\varrho_{\rm M}\,,\\
[\nabla\cdot\myvec H_2]^\pm&=
\pm\frac{{\rm i}k_1}{2}\sigma_{\rm M}
-\frac{1}{2}\nabla\cdot\tilde{\cal S}_{k_2}(\myvec\nu'\sigma_{\rm M}
+\kappa\myvec M_{\rm s})
+\frac{{\rm i}k_1}{2}\kappa\tilde{\cal S}_{k_2}\varrho_M\,.
\end{align}

\renewcommand{\theequation}{B.\arabic{equation}}
\setcounter{equation}{0}

\subsection*{B. Divergence conditions}

The derivations of the conditions for \eqref{eq:divEW},
\eqref{eq:divHW}, and \eqref{eq:divH0E0} to hold are all very similar.
For this reason we only present a detailed derivation of the condition
for \eqref{eq:divHW} to hold.

The fields $\myvec E_W$ and $\myvec H_{W}$ are defined through
\eqref{eq:EHW1}, \eqref{eq:E1rep1}--\eqref{eq:H2rep11}, and the
solution to \eqref{eq:EHsys}. Appendix A and \eqref{eq:EHsys} give the
relations on $\Gamma$
\begin{align}
\lambda\kappa\myvec \nu\times\myvec E_W^+&=
\myvec\nu\times\myvec E_W^-\,,\label{eq:RVA1}\\
\lambda\kappa\myvec \nu\cdot\myvec H_W^+&=
\myvec\nu\cdot\myvec H_W^-\,,\label{eq:RVA2}\\
\gamma_{\rm M}[\nabla\cdot\myvec H_W]^+&=
[\nabla\cdot\myvec H_W]^-\,.\label{eq:RVA3}
\end{align}
By combining the surface divergence of
\eqref{eq:RVA1} with \eqref{eq:RVA2} we get
\begin{equation}
\lambda\kappa({\rm i}k_1\myvec\nu\cdot\myvec
H_2^+-\myvec\nu\cdot[\nabla\times\myvec E_2)]^+)=
{\rm i}k_1\myvec\nu\cdot\myvec H_1^-
-\myvec\nu\cdot[\nabla\times \myvec E_1)]^-\,,
\end{equation}
where we have used $\myvec\nu\cdot(\nabla\times\myvec\nu\times(\myvec\nu\times
\myvec E_i))=-\myvec \nu\cdot(\nabla\times\myvec E_i)$, $i=1,2$.
By \eqref{eq:E1rep1}--\eqref{eq:H2rep11} and limits in Appendix
A this leads to
\begin{multline}
\lambda\kappa\left(\kappa^{-1}\myvec\nu\cdot
[\nabla(\nabla\cdot{\cal S}_{k_2})]^+(\myvec\nu'\sigma_{\rm M}
+\kappa\myvec M_{\rm s})
-{\rm i}k_1\myvec\nu\cdot\boldsymbol{\cal N}_{k_2}\varrho_{\rm M}
+{\rm i}k_1\varrho_{\rm M}\right)\\
=-\myvec\nu\cdot[\nabla(\nabla\cdot {\cal S}_{k_1})]^-
(\myvec\nu'\sigma_{\rm M}+\myvec M_{\rm s})
+{\rm i}k_1\myvec\nu\cdot\boldsymbol{\cal N}_{k_1}\varrho_{\rm M}
+{\rm i}k_1\varrho_{\rm M}\,.
\label{eq:surfdiv}
\end{multline}
A comparison of \eqref{eq:surfdiv} with the limits $\myvec
\nu\cdot[\nabla(\nabla\cdot\myvec H_1)]^-$ and $\myvec
\nu\cdot[\nabla(\nabla\cdot\myvec H_2)]^+$ gives
\begin{equation}
\lambda\myvec\nu\cdot[\nabla(\nabla\cdot\myvec H_2)]^+=
\myvec\nu\cdot[\nabla(\nabla\cdot\myvec H_1)]^-\,.
\label{eq:divcond2}
\end{equation} 
Let $\psi_W=\nabla\cdot \myvec H_W$, with $\myvec H_W$ from
\eqref{eq:EHW1}. The fundamental solution \eqref{eq:fund} and the
boundary conditions \eqref{eq:RVA3} and \eqref{eq:divcond2} make
$\psi_W$ satisfy
\begin{equation}
\left\{
\begin{array}{ll}
\Delta\psi_W(\myvec r)+k_2^2\psi_W(\myvec r)=0\,,& \myvec r\in\Omega_1\,,\\
\Delta\psi_W(\myvec r)+k_1^2\psi_W(\myvec r)=0\,,& \myvec r\in\Omega_2\,,\\
\gamma_{\rm M} \psi_W^+(\myvec r)=\psi_W^-(\myvec r)\,,& \myvec r\in\Gamma\,,\\
\lambda\myvec\nu\cdot[\nabla\psi_W]^+(\myvec r)=
\myvec\nu\cdot[\nabla\psi_W]^-(\myvec r)\,,& \myvec r\in\Gamma\,,\\
\left(\partial_{\hat{\myvec r}}-{\rm i}k_2\right)\psi_W(\myvec r)=
o\left(\vert\myvec r\vert^{-1}\right)\,,& \vert\myvec r\vert\to\infty\,.
\end{array}\right.
\label{eq:psiWeq}
\end{equation}
By rescaling $\psi_W$ in $\Omega_1$, problem \eqref{eq:psiWeq} becomes
identical to problem ${\sf B}_0$ with $\alpha=\lambda\bar{\gamma}_{\rm
  M}/\vert{\gamma}_{\rm M}\vert^2$. Thus if
$\{k_1,k_2,\alpha=\lambda\bar{\gamma}_{\rm M}\}$ is such that the
conditions of Proposition~\ref{prop:uniqueB0} hold, then
\eqref{eq:psiWeq} only has the trivial solution $\nabla\cdot\myvec
H_W=0$ for $\myvec r\in\Omega_1\cup\Omega_2$.

The condition for $\nabla\cdot\myvec E_W=0$ is that the set
$\{k_1,k_2,\alpha=\bar\gamma_{\rm E}\bar\kappa\}$ is such that the
conditions of Proposition~\ref{prop:uniqueB0} hold. The condition for
\eqref{eq:divH0E0} to hold is that $(\Arg(k_1),\Arg(k_2))$ is in the
set of points of Figure \ref{fig:hexagon}(a).

\renewcommand{\theequation}{C.\arabic{equation}}
\setcounter{equation}{0}

\subsection*{C. Fulfillment of Maxwell's equations}

We show that $\myvec E$ and $\myvec H$ of \eqref{eq:repEH} satisfy
\eqref{eq:Max123C} and that $\myvec E_W$ and $\myvec H_W$ of
\eqref{eq:EHW1} satisfy \eqref{eq:D0123} if $\nabla\cdot\myvec
E_i(\myvec r)=\nabla\cdot\myvec H_i(\myvec r)=0$, $i=1,2$, and $\myvec
r\in\Omega_1\cup\Omega_2$.

The rotation of \eqref{eq:H1rep11} and \eqref{eq:H2rep11} can be
written
\begin{align}
\nabla\times\myvec H_1(\myvec r)&=
\frac{{\rm i}k_1}{2}{\cal R}_{k_1}(\myvec\nu'\sigma_{\rm M}
+\myvec M_{\rm s})(\myvec r)
-\frac{{\rm i}k_1}{2}\tilde{\cal S}_{k_1}(\myvec\nu'\sigma_{\rm E}
+\myvec J_{\rm s})(\myvec r)
\nonumber\\
+\frac{1}{2}\nabla&(\nabla\cdot{\cal S}_{k_1}(\myvec\nu'\sigma_{\rm E}
+\myvec J_{\rm s}))(\myvec r)+\nabla\times \myvec H^{\rm in}(\myvec r)\,,
\quad\myvec r\in \Omega_1\cup \Omega_2\,,
\label{eq:H1rep11NN}\\
\nabla\times\myvec H_2(\myvec r)&=
-\frac{{\rm i}k_1}{2}{\cal R}_{k_2}(\myvec\nu'\sigma_{\rm M}
+\kappa\myvec M_{\rm s})(\myvec r)
+\frac{{\rm i}k_1}{2}\tilde{\cal S}_{k_2}
(\myvec\nu'\sigma_{\rm E}+\kappa\myvec J_{\rm s})(\myvec r)
\nonumber\\
&-\frac{1}{2}\nabla(\nabla\cdot{\cal S}_{k_2}
(\kappa^{-1}\myvec\nu'\sigma_{\rm E}+\myvec J_{\rm s}))(\myvec r)\,,
\quad\myvec r\in\Omega_1\cup\Omega_2\,.
\label{eq:H2rep11NN}
\end{align}
If $\nabla\cdot\myvec E_i=0$, $i=1,2$, it follows from
\eqref{eq:E1rep1} and \eqref{eq:E2rep1} that
\begin{align}
\tilde{\cal S}_{k_1}\varrho_{\rm E}(\myvec r)
-\nabla\cdot{\cal S}_{k_1}(\myvec \nu'\sigma_{\rm E}
+\myvec J_{\rm s})(\myvec r)&=0\,,
\quad\myvec r\in \Omega_1\cup\Omega_2\,,\label{eq:E1rep1mod}\\
\tilde{\cal S}_{k_2}\varrho_{\rm E}(\myvec r)
-\nabla\cdot{\cal S}_{k_2}(\kappa^{-1}\myvec\nu'\sigma_{\rm E}
+\myvec J_{\rm s})(\myvec r)&=0\,,
\quad\myvec r\in\Omega_1\cup\Omega_2\,.\label{eq:E2rep1mod}
\end{align}
The Amp\`ere law  
\begin{equation}
\begin{split}
\nabla\times\myvec H_1(\myvec r)&=-{\rm i}k_1\myvec E_1(\myvec r)\,,
\quad\myvec r\in\Omega_1\cup\Omega_2\,,\\
\nabla\times\myvec H_2(\myvec r)&=-{\rm i}k_1\kappa\myvec E_2(\myvec r)\,,
\quad\myvec r\in\Omega_1\cup\Omega_2\,,
\end{split}
\label{eq:Max24}
\end{equation}
now follows by combining \eqref{eq:H1rep11NN} and \eqref{eq:E1rep1mod}
with \eqref{eq:E1rep1}, and by combining \eqref{eq:H2rep11NN} and
\eqref{eq:E2rep1mod} with \eqref{eq:E2rep1}. The Faraday law
\begin{equation}
\nabla\times\myvec E_i(\myvec r)={\rm i}k_1\myvec H_i(\myvec r)\,,
\quad i=1,2\,,
\quad\myvec r\in\Omega_1\cup\Omega_2\,,
\label{eq:Max13}
\end{equation}
follows in the same manner from $\nabla\cdot\myvec H_i=0$, $i=1,2$,
and by combining the rotation of \eqref{eq:E1rep1} with
\eqref{eq:H1rep11} and the rotation of \eqref{eq:E2rep1} with
\eqref{eq:H2rep11}. From \eqref{eq:Max24} and \eqref{eq:Max13} it
follows that $\myvec E$ and $\myvec H$ of \eqref{eq:repEH} satisfy
\eqref{eq:Max123C} and that $\myvec E_W$ and $\myvec H_W$ of
\eqref{eq:EHW1} satisfy \eqref{eq:D0123}.

\renewcommand{\theequation}{D.\arabic{equation}}
\setcounter{equation}{0}

\subsection*{D. Uniqueness for problems {\sf C}, ${\sf C}_0$, and 
  ${\sf D}_0$}

We sketch a proof that problem ${\sf C}_0$ has only the trivial
solution and that problem {\sf C} has at most one solution by relating
these problems to problem ${\sf A}_0$ and {\sf A}. We also justify
that the criteria for problem ${\sf D}_0$ to only have the trivial
solution are the same as the criteria in
Proposition~\ref{prop:uniqueB0} that make problem ${\sf B}_0$ only
have the trivial solution.

Let $S_R$ be a sphere of radius $R$ with outward unit normal $\myvec
n$. Assume that $S_R$ is sufficiently large to contain $\Gamma$ and
let $\Omega_{1,R}=\{\myvec r\in\Omega_1:\vert\myvec r\vert < R\}$.
From Gauss' theorem we obtain energy relations for problem ${\sf A}_0$
and problem ${\sf C}_0$
\begin{equation}
\int_{S_R}(U\nabla\bar U)\cdot\myvec n\,{\rm d}S=
\int_{\Omega_{1,R}}\left(\vert\nabla U\vert^2
-\bar k_1^2\vert U\vert^2\right){\rm d}v+
\int_{\Omega_{2}}\left(\bar\kappa^{-1}\vert\nabla U\vert^2-\bar k_1^2\vert
U\vert^2\right){\rm d}v\,,
\label{eq:unablau}
\end{equation}
\begin{multline}
-{\rm i}\bar k_1\int_{S_R}(\bar{\myvec E}\times{\myvec H})\cdot\myvec n
\,{\rm d}S=
\int_{\Omega_{1,R}}\left(\vert k_1\vert^2\vert\myvec E\vert^2-\bar k_1^2\vert 
\myvec H\vert^2\right){\rm d}v\\
+\int_{\Omega_{2}}\left(\vert k_1\kappa\vert^2\bar{\kappa}^{-1}
\vert\myvec E\vert^2-\bar k_1^2\vert\myvec H\vert^2\right){\rm d}v\,.
\label{eq:barEH}
\end{multline}
The right hand sides of \eqref{eq:unablau} and \eqref{eq:barEH} are
equivalent. By using techniques similar to those in
\cite[pp.~309--310]{KleiMart88} and \cite[p.~1434]{KresRoac78} it
follows that when $(\Arg(k_1),\Arg(k_2))$ is in the set of points of
Figure \ref{fig:hexagon}(a), then $U=0$ and $\myvec E=\myvec H=\myvec
0$ in $\Omega_1\cup \Omega_2$. Standard arguments give that problem
{\sf C} has at most one solution when problem ${\sf C}_0$ only has the
trivial solution.

In the same manner as above the energy relation for problem ${\sf
  D}_0$ is shown to be equivalent to the energy relation for problem
${\sf B}_0$. We can again use \cite[pp.~309--310]{KleiMart88} and
\cite[p.~1434]{KresRoac78} to find the criteria that lead to $W=0$ and
$\myvec H_W=\myvec E_W=\myvec 0$. These are the criteria for the set
$\{k_1,k_2,\alpha=\lambda\}$ in Proposition~\ref{prop:uniqueB0}.

\begin{small}

\end{small}


\begin{thebibliography}{00}
  
\bibitem{Axels06} A. Axelsson, ``Transmission problems for Maxwell's
  equations with weakly Lipschitz interfaces'', Math. Methods Appl.
  Sci., {\bf 29}, 665--714 (2006).
   
\bibitem{Annsop16} A.-S. Bonnet-Ben Dhia and C. Carvalho and L.
  Chesnel and P. Ciarlet, ``On the use of perfectly matched layers at
  corners for scattering problems with sign-changing coefficients'',
  J. Comput. Phys., {\bf 322}, 224--247 (2016).

\bibitem{ColtKres83} D. Colton and R. Kress, \emph{Integral equation
    methods in scattering theory}, John Wiley \& Sons Inc., New York,
  1983.
  
\bibitem{ColtKres98} D. Colton and R. Kress, \emph{Inverse acoustic and
    electromagnetic scattering theory}, 2nd ed., Appl. Math. Sci.,
  vol. 93, Springer-Verlag, Berlin, 1998.

\bibitem{EpsGreNei13} C. L. Epstein, L. Greengard, and M. O'Neil,
  ``Debye sources and the numerical solution of the time harmonic
  {M}axwell equations II'', Comm. Pure Appl. Math., {\bf 66},
  753--789 (2013).
   
\bibitem{EpsGreNei19} C. L. Epstein, L. Greengard, and M. O'Neil, ``A
  high-order wideband direct solver for electromagnetic scattering
  from bodies of revolution'', J. Comput. Phys., {\bf 387}, 205--229
  (2019).
  
\bibitem{Hels09} J. Helsing, ``Integral equation methods for elliptic
  problems with boundary conditions of mixed type'', J. Comput.
  Phys., {\bf 228}, 8892--8907 (2009).
  
\bibitem{Hels18} J. Helsing, ``Solving integral equations on
  piecewise smooth boundaries using the RCIP method: a tutorial'',
  arXiv:1207.6737v9 [physics.comp-ph] (revised 2018).

\bibitem{HelsKarl17} J. Helsing and A. Karlsson, ``Resonances in
  axially symmetric dielectric objects'', IEEE Trans. Microw. Theory
  Tech., {\bf 65}, 2214--2227 (2017).
  
\bibitem{HelsKarl18} J. Helsing and A. Karlsson, ``On a Helmholtz
  transmission problem in planar domains with corners'', J. Comput.
  Phys., {\bf 371}, 315--332 (2018).
  
\bibitem{HelsKarl19} J. Helsing and A. Karlsson, ``Physical-density
  integral equation methods for scattering from multi-dielectric
  cylinders'', J. Comput. Phys., {\bf 387}, 14--29 (2019).
  
\bibitem{HelsPerf18} J. Helsing and K.-M. Perfekt, ``The spectra of
  harmonic layer potential operators on domains with rotationally
  symmetric conical points'', J. Math. Pures Appl., {\bf 118},
  235--287 (2018).
  
\bibitem{HelsRose20} J. Helsing and A. Rosén, ``Dirac integral
  equations for dielectric and plasmonic scattering'',
  arXiv:1911.00788 [math.AP] (2019).
  
\bibitem{Homola08} J. Homola, ``Surface plasmon resonance sensors for
  detection of chemical and biological species'', Chem. Rev., {\bf
    108}, 462--493 (2008).

\bibitem{KirsHett15} A. Kirsch and F. Hettlich, \emph{The mathematical
    theory of time-harmonic Maxwell's equations}, Appl. Math. Sci.,
  vol. 190, Springer, Cham, 2015.
  
\bibitem{KleiMart88} R.E. Kleinman and P.A. Martin, ``On single
  integral equations for the transmission problem of acoustics'', SIAM
  J. Appl. Math., {\bf 48}, 307--325 (1988).
  
\bibitem{KresRoac78} R. Kress and G.F. Roach. ``Transmission problems
  for the Helmholtz equation'', J. Math. Phys., {\bf 19}, 1433--1437
  (1978).
    
\bibitem{LaiOneil19} J. Lai and M. O'Neil, ``An FFT-accelerated direct
  solver for electromagnetic scattering from penetrable axisymmetric
  objects'', J. Comput. Phys., {\bf 390}, 152--174 (2019).
  
\bibitem{LiFuShank18} J. Li, X. Fu, and B. Shanker, ``Decoupled
  potential integral equations for electromagnetic scattering from
  dielectric objects'', IEEE Trans. Antennas Propag., {\bf 67},
  1729--1739 (2018).
  
\bibitem{LuTsGuHo19} X. Luo, D. Tsai, M. Gu, and M. Hong,
  ``Extraordinary optical fields in nanostructures: from
  sub-diffraction-limited optics to sensing and energy conversion''
  Chem. Soc. Rev., {\bf 48}, 2458--2494 (2019).
  
\bibitem{MautHarr77} J.R. Mautz and R.F. Harrington.
  ``Electromagnetic scattering from a homogeneous body of
  revolution'', Tech. Rep. TR-77-10, Dept. of electrical and computer
  engineering, Syracuse Univ., New York, (1977).
  
\bibitem{Muller69} C. Müller, {\it Foundations of the Mathematical
    Theory of Electromagnetic Waves.} Berlin, Springer-Verlag, 1969.
  
\bibitem{Raeth88} H. Raether, {\it Surface plasmons on smooth and rough
    surfaces and on gratings}, vol. 111 of {\it Springer tracts in
    modern physics}, Springer, Berlin, 1988.
  
\bibitem{TaskOija06} M. Taskinen and P. Ylä-Oijala, ``Current and
  charge integral equation formulation'', IEEE Trans. Antennas
  Propag., {\bf 54}, 58--67 (2006).
  
\bibitem{TzarSihv18} D. Tzarouchis and A. Sihvola, ``Light scattering
  by a dielectric sphere: perspectives on the Mie resonances'', Appl.
  Sci., {\bf 8}, Article no. 184 (2018).
  
\bibitem{Vico16} F. Vico, M. Ferrando-Bataller, T. Jiménez, and D.
  Sánchez-Escuderos, ``A non-resonant single source augmented integral
  equation for the scattering problem of homogeneous lossless
  dielectrics'', 2016 IEEE Int. Symp. Antennas Propag. (APSURSI),
  745--746 (2016).
  
\bibitem{VicGreFer18} F. Vico, L. Greengard, and M. Ferrando.
  ``Decoupled field integral equations for electromagnetic scattering
  from homogeneous penetrable obstacles'', Comm. Part. Differ. Equat.
  {\bf 43}, 159--184 (2018).
  
\bibitem{YouHaoMar12} P. Young, S. Hao, and P.G. Martinsson, ``A
  high-order Nyström discretization scheme for boundary integral
  equations defined on rotationally symmetric surfaces'', J. Comput.
  Phys., {\bf 231}, 4142--4159 (2012).

\end{thebibliography}
\end{document}